\documentclass[10pt]{article}
\usepackage[T1]{fontenc}
\usepackage{ae,aecompl}
\usepackage[utf8x]{inputenc}
\usepackage[cm]{fullpage}
\usepackage{mdwlist}
\usepackage[small]{caption}
\usepackage[numbers,sort&compress,square]{natbib}
\usepackage{ctable}
\makeatletter
\renewcommand\bibsection%
{
  \section*{\refname
    \@mkboth{\MakeUppercase{\refname}}{\MakeUppercase{\refname}}}
}
\makeatother

\usepackage{times}
\usepackage{amsfonts}
\usepackage{amsmath}
\usepackage{amssymb}
\usepackage{amsthm}
\usepackage{mathtools}
\usepackage{graphicx}
\usepackage{subfig}
\usepackage{multirow}
\usepackage{bigstrut}
\usepackage{url}
\usepackage[ruled,linesnumbered]{algorithm2e}
\usepackage{mathptmx}      % use Times fonts if available on your TeX system
\DeclareMathAlphabet{\mathcal}{OMS}{cmsy}{m}{n} % Use standard fonts for calligraphic

\newcommand{\Sam}{{\cal S}}
\newcommand{\Ds}{{\cal D}}
\def\Itm{{\cal I}}
\def\TOPK{\mathsf{TOPK}}
\def\FI{\mathsf{FI}}
\def\AR{\mathsf{AR}}
\def\VC{\mathsf{VC}}

\newtheorem{corollary}{Corollary}

\newtheorem{lemma}{Lemma}

\newtheorem{theorem}{Theorem}

\theoremstyle{definition}
\newtheorem{definition}{Definition}

\begin{document}
\title{Efficient Discovery of Association Rules and Frequent
Itemsets through Sampling with Tight Performance Guarantees\thanks{Work was
supported in part by NSF award IIS-0905553.}\ \footnote{A shorter version of this
paper appeared in the proceedings of ECML PKDD 2012 as~\citep{RiondatoU12}.}}
\author{Matteo Riondato\footnote{Contact author.} and Eli Upfal \\
Department of Computer Science, Brown University, Providence,
RI, USA \\ {\texttt \{matteo, eli\}@cs.brown.edu}}

\maketitle

\section{Introduction}\label{sec:intro}
Discovery of frequent itemsets and association rules is a fundamental
computational primitive with application in data mining (market basket
analysis), databases (histogram construction), networking (heavy hitters) and
more~\cite[Sect.~5]{HanCXY07}. Depending on the particular application, one is
interested in finding all itemsets with frequency greater or equal to a user
defined threshold (FIs), identifying the $K$ most frequent itemsets (top-$K$),
or computing all association rules (ARs) with user defined minimum  support and
confidence level. Exact solutions to these problems require scanning the entire
dataset, possibly multiple times. For large datasets that do not fit in main
memory, this can be prohibitively expensive. Furthermore, such extensive
computation is often unnecessary, since high quality approximations are
sufficient for most practical applications.  Indeed, a number of recent
papers (see~\ref{sec:prevwork}) for more details)
%~\citep{MannilaTV94,Toivonen96,ZakiPLO97,JiaL05,LiG04,ZhangZW03,ZhaoZZ06,BronnimanCDHS03,ChandraB11,ChenHS02,ChenHH11,ChuangCY05,HuY06,HwangK06,JiaG05,JohnL96,MahafzahABAZ09,Parthasarathy02,WangDC05,ChuangHC08,ChakaravarthyPS09,PietracaprinaV07,WangHLT05,PietracaprinaRUV10,SchefferW02,VasudevanV09}
explored the application of sampling for approximate solutions to these
problems. However, the efficiency and practicality of the sampling approach
depends on a tight relation between the size of the sample and the quality of
the resulting approximation. Previous works do not provide satisfactory
solutions to this problem.

The technical difficulty in analyzing any sampling technique for frequent
itemset discovery problems is that a-priori any subset of items can be among
the most frequent ones, and the number of subsets is exponential in the number
of distinct items appearing in the dataset. A standard analysis begins with a bound on
the probability that a given itemset is either over or under represented in the
sample. Such bound is easy to obtain using a large deviation bound such as the
Chernoff bound or the Central Limit theorem~\citep{MitzenmacherU05}. The
difficulty is in combining the bounds for individual itemsets
into a global bound that holds simultaneously for all the itemsets. A simple
application of the union bound vastly overestimates the error probability
because of the large number of possible itemsets, a large fraction of which may
not be present in the dataset and therefore should not be considered. More
sophisticated techniques, developed in recent
works~\citep{ChakaravarthyPS09,PietracaprinaRUV10,ChuangCY05}, give better
bounds only in limited cases. A loose bound 
 on the required sample size for achieving the user defined
 performance guarantees, decreases the gain obtained from the use of sampling. 

In this work we circumvent this problem 
through a novel application of the \emph{Vapnik-Chervonenkis (VC)} dimension
concept, a fundamental tool in statistical learning theory.  Roughly speaking,
the VC-dimension of a collection of indicator functions (a range space) is a
measure of its complexity or expressiveness (see Sect.~\ref{sec:prelvcdim} for
formal definitions). A major result~\citep{VapnikC71} relates the VC-dimension of
a range space to a sufficient size for a random sample to simultaneously
approximate all the indicator functions within predefined parameters. The main
obstacle in applying the VC-dimension theory to particular computation problems
is computing the VC-dimension of the range spaces associated with these
problems.  

We apply the VC-dimension theory to frequent itemsets problems by viewing the
presence of an itemset in a transaction as the outcome of an indicator function
associated with the itemset. The major theoretical contributions of our work are
a complete characterization of the VC-dimension of the range space associated
with a dataset, and a tight bound to this quantity. We prove that the VC-dimension
is upper bounded by a
%n easy-to-compute 
characteristic quantity of the dataset
which we call \emph{d-index}. The d-index is the maximum integer $d$ such that the
dataset contains at least $d$ different transactions of length at least $d$ such
that no one of them is a subset of or equal to another in the considered set
of transactions (see Def.~\ref{defn:dindex}). We show that this bound is tight
by
demonstrating a large class of datasets with a VC-dimension that matches the
bound. Computing the d-index can be done in polynomial time but it requires
multiple scans of the dataset. We show how to compute an upper bound to the
d-index with a single linear scan of the dataset in an online greedy fashion.

The VC-dimension approach provides a unified tool for analyzing the various
frequent itemsets and association rules problems (i.e., the market basket
analysis tasks). We use it to prove tight bounds on the required
sample size for extracting FI's with a minimum frequency threshold, for mining
the top-$K$ FI's, and for computing the collection of AR's with minimum
frequency and confidence thresholds. Furthermore, we compute bounds for both
absolute and relative approximations (see Sec.~\ref{sec:preldm} for definitions).
We show that high quality approximations can be obtained by mining a very small
random sample of the dataset. Table~\ref{table:comparsamsizeform} compares our
technique to the best previously known results for the various problems (see
Sect.~\ref{sec:preldm} for definitions). Our bounds, which are linear in the
VC-dimension associated with the dataset, are consistently smaller than previous
results and less dependent on other parameters of the problem such as the
minimum frequency threshold and the dataset size. An extensive
experimental evaluation demonstrates the advantage of our technique in practice.

This work is the first to provide a characterization and an explicit bound for
the VC-dimension of the range space associated with a dataset and to apply the
result to the extraction of FI's and AR's from random sample of the dataset. We
believe that this connection with statistical learning theory can be furtherly
exploited in other data mining problems.

\ctable[
	cap     = {Comparison of sample sizes with previous works},
	caption = {Required sample sizes (as number of transactions) for various
	approximations to FI's and AR's as functions of
  the VC-dimension $d$, the maximum transaction length $\Delta$, the number of
  items $|\Itm|$, the accuracy $\varepsilon$, the failure probability $\delta$,
  the minimum frequency $\theta$, and the minimum confidence $\gamma$. Note that
  $d\leq \Delta\leq |\Itm|$ (but $d<|\Itm|)$. $c$ and $c'$ are absolute constants.
  with $c\le 0.5$.},
	label   = {table:comparsamsizeform},
	pos = tphb,
]{lll}{
	\tnote[$\dag$]{\citep{Toivonen96,JiaL05,LiG04,ZhangZW03}}
	\tnote[$\ddag$]{\citep{ChakaravarthyPS09}}
	\tnote[$\S$]{\citep{SchefferW02,PietracaprinaRUV10}}
	\tnote[$\P$]{\citep{ChakaravarthyPS09}}
	}{ \FL
    \toprule
    Task/Approx. & This work & Best previous work \ML
    FI's/abs. & $\frac{4c}{\varepsilon^2}\left(d+\log\frac{1}{\delta}\right)$&
    $O\left(\frac{1}{\varepsilon^2}\left(|\Itm|+\log\frac{1}{\delta}\right)\right)$\tmark[$\dag$]\bigstrut \NN
  FI's/rel. &
  $\frac{4(2+\varepsilon)c}{\varepsilon^2(2-\varepsilon)\theta}\left(d\log\frac{2+\varepsilon}{\theta(2-\varepsilon)}+\log\frac{1}{\delta}\right)$
  & $\frac{24}{\varepsilon^2(1-\varepsilon)\theta}\left(\Delta +5
  +\log\frac{4}{(1-\varepsilon)\theta\delta}\right)$\tmark[$\ddag$] \bigstrut \NN
  top-$K$ FI's/abs. & $\frac{16c}{\varepsilon^2}\left(d+\log\frac{1}{\delta}\right)$ &
  $O\left(\frac{1}{\varepsilon^2}\left(|\Itm|+\log\frac{1}{\delta}\right)\right)$\tmark[$\S$]\bigstrut \NN
  top-$K$ FI's/rel. &
  $\frac{4(2+\varepsilon)c'}{\varepsilon^2(2-\varepsilon)\theta}\left(d\log\frac{2+\varepsilon}{\theta(2-\varepsilon)}+\log\frac{1}{\delta}\right)$
  & not available \bigstrut \NN
  AR's/abs. &
  $O\left(\frac{(1+\varepsilon)}{\varepsilon^2(1-\varepsilon)\theta}\left(d\log\frac{1+\varepsilon}{\theta(1-\varepsilon)}+\log\frac{1}{\delta}\right)\right)$
  & not available \bigstrut \NN
  AR's/rel. &
  $\frac{16c'(4+\varepsilon)}{\varepsilon^2(4-\varepsilon)\theta}\left(d\log\frac{4+\varepsilon}{\theta(4-\varepsilon)}+\log\frac{1}{\delta}\right)$
  & $\frac{48}{\varepsilon^2(1-\varepsilon)\theta}\left(\Delta +5
  +\log\frac{4}{(1-\varepsilon)\theta\delta}\right)$\tmark[$\P$]
  \bigstrut \LL
}

%\subsection{Outline} 
\paragraph{Outline}
We review relevant previous work in Sect.~\ref{sec:prevwork}. In
Sect.~\ref{sec:prelim} we formally define the problem
and our goals, and introduce definitions and lemmas used in the analysis. The
main part of the analysis with derivation of a strict bound to the VC-dimension
of association rules is presented in Sect.~\ref{sec:vcdimar}, while our
algorithms and sample sizes for mining FI's, top-$K$ FI's, and association rules
through sampling are in Sect.~\ref{sec:approx}. Section~\ref{sec:exp} contains
an extensive experimental evaluation of our techniques. A discussion of our
results and the conclusions can be found in Sect.~\ref{sec:concl}.

\section{Related Work}\label{sec:prevwork}
\citet{AgrawalIS93} introduced the problem of mining association
rules in the basket data model, formalizing a fundamental task of information
extraction in large datasets. Almost any known algorithm for the problem starts
by solving a FI's problem and then generate the association rules implied by
these frequent itemsets. \citet{AgrawalS94} presented
\emph{Apriori}, the most well-known algorithm for mining FI's, and
\emph{FastGenRules} for computing association rules from a set of itemsets.
Various ideas for improving the efficiency of FI's and AR's algorithms have been
studied, and we refer the reader to the survey by~\citet{CeglarR06} for a good presentation of recent contributions.
However, the running times of all known algorithms heavily depend on the size of
the dataset.  

\citet{MannilaTV94} were the first to suggest the 
use of sampling to efficiently identify the collection of FI's, presenting some empirical
results to validate the intuition. \citet{Toivonen96} presents an
algorithm that, by mining a random sample of the dataset, builds a candidate set
of frequent itemsets which contains all the frequent itemsets with a probability
that depends on the sample size. There are no guarantees that all itemsets
in the candidate set are frequent, but the set of candidates can be used to
efficiently identify the set of frequent itemsets with at most two passes over
the entire dataset. This work also suggests a bound on the sample size sufficient
to ensure that the frequencies of itemsets in the sample are close to their real
one. The analysis uses Chernoff bounds and the union bound. The major drawback
of this sample size is that it depends linearly on the number of individual
items appearing in the dataset.

\citet{ZakiPLO97} show that static sampling is an efficient way to
mine a dataset, but choosing the sample size using Chernoff bounds is too
conservative, in the sense that it is possible to obtain the same accuracy and
confidence in the approximate results at smaller sizes than what the theoretical
analysis suggested. 

Other works tried to improve the bound to the sample size by using different
techniques from statistics and probability theory like the central limit
theorem~\citep{ZhangZW03,LiG04,JiaL05} or hybrid Chernoff
bounds~\citep{ZhaoZZ06}.

Since theoretically-derived bounds to the sample size where too loose to be
useful, a corpus of works applied progressive sampling to extract
FI's~\citep{JohnL96,ChenHS02,Parthasarathy02,BronnimanCDHS03,ChuangCY05,JiaG05,WangDC05,HwangK06,HuY06,MahafzahABAZ09,ChenHH11,ChandraB11}.
Progressive sampling algorithms work by selecting a random sample and then
trimming or enriching it by removing or adding new sampled transactions
according to a heuristic or a self-similarity measure that is fast to evaluate,
until a suitable stopping condition is satisfied. The major downside of this
approach is that it offers no guarantees on the quality of the obtained results.

Another approach to estimating the required sample size is presented
by~\citet{ChuangHC08}. The authors give an algorithm that studies the
distribution of frequencies of the itemsets and uses this information to fix a
sample size for mining frequent itemsets, but without offering any theoretical
guarantee.

A recent work by \citet{ChakaravarthyPS09} gives the first
analytical bound on a sample size that is linear in the length of the longest
transaction, rather than in the number of items in the dataset.  This work is
also the first to present an algorithm that uses a random sample of the dataset
to mine approximated solutions to the AR's problem with quality guarantees. No
experimental evaluation of their methods is presented, and they do not address
the top-K FI's problem. Our approach gives better bounds for the problems
studied in~\citep{ChakaravarthyPS09} and applies to related problems such as the
discovery of top-$K$ FI's and absolute approximations.

Extracting the collection of top-$K$ frequent itemsets is a more difficult task
since the corresponding minimum frequency threshold is not known in
advance~\citep{CheungF04,FuKT00}. Some works solved the problem by looking at
\emph{closed} top-$K$ frequent itemsets, a concise representation of the
collection~\citep{WangHLT05,PietracaprinaV07}, but they suffers from the same
scalability problems as the algorithms for exactly mining FI's with a fixed
minimum frequency threshold.

Previous works that used sampling to approximation the collection of top-$K$
FI's~\citep{SchefferW02,PietracaprinaRUV10} used progressive sampling. Both
works provide (similar) theoretical guarantees on the quality of the
approximation. What is more interesting to us, both works present a theoretical
upper bound to the sample size needed to compute such an approximation. The size
depended linearly on the number of items. 
In contrast, our results give a sample size that only in the worst case is
linear in the number of items but can be (and is, in practical cases) much less
than that, depending on the dataset, a flexibility not provided by previous
contributions. 
Sampling is used by \citet{VasudevanV09} to extract
an approximation of the top-$K$ frequent individual \emph{items} from a sequence
of items, which contains no item whose actual frequency is less than
$f_K-\varepsilon$ for a fixed $0<\varepsilon<1$, where $f_K$ is the
\emph{actual} frequency of the $K$-th most frequent item. They derive a sample
size sufficient to achieve this result, but they assume the knowledge of $f_K$,
which is rarely the case. An empirical sequential method can be used to estimate
the right sample size. Moreover, the results cannot be directly extended to the
mining of top-$K$ frequent item(set)s from datasets of transactions with length
greater than one.

The {\em Vapnik-Chervonenkis dimension} was first introduced in a seminal
article~\citep{VapnikC71} on the convergence of probability distributions, but it
was only with the work of~\citet{HausslerW86} and~\citet{BlumerEHW89} that it
was applied to the field of learning. \citet{BoucheronBL05} present a good survey
of the field with many recent advances. Since then, VC-dimension has encountered
enormous success and application in the fields of computational
geometry~\citep{Chazelle00,Matousek02} and machine
learning~\citep{AnthonyB99,DevroyeGL96}. Other applications include
database management and graph algorithms. 
In the former, it was used in the
context of constraint databases to compute good approximations of aggregate
operators~\citep{BenediktL02}. VC-dimension-related
results were also recently applied in the field of database privacy
by~\citet{BlumLR08} to show a bound on the number of queries
needed for an attacker to learn a private concept in a database. \citet{Gross11}
showed that content with unbounded
VC-dimension can not be watermarked for privacy purposes.
\citet{RiondatoACZU11} computed an upper bound to the VC-dimension of
classes of SQL queries and used it to develop a sampling-based algorithm for
estimating the size of the output (selectivity) of queries run on a dataset.
The results therein, although very different from what presented here due to the
different settings, the different goals, and the different techniques used,
inspired our present work. 
In the graph algorithms literature, VC-Dimension has been used to develop
algorithms to efficiently detect network
failures~\citep{Kleinberg03,KleinbergSS08}, balanced separators~\citep{FeigeM06},
events in a sensor networks~\citep{GandhiSW10}, and compute the shortest
path~\citep{AbrahamDFGW11}. To our knowledge, this work is the
first application of VC-dimension to knowledge discovery.

In this present article we extend our previous published work~\citep{RiondatoU12}
in a number of ways. The most prominent change is the development and analysis
of a tighter bound to the VC-dimension of the range space associated to the
dataset, together with a new polynomial time algorithm to compute such bound and
a very fast linear time algorithm to compute an upper bound. The proofs to most
of our results were not published in the conference version but are presented
here. We also added numerous examples to improve the understanding of the
definitions and of the theoretical results, and explained the connection of our
results with other known results in statistical learning theory. As far as the
experimental evaluation is concerned, we added comments on the precision and
recall of our methods and on their scalability, which is also evident from their
use inside a parallel/distributed algorithm for FI's and AR's
mining~\citep{RiondatoDFU12} for the MapReduce~\citep{DeanS04} platform that we
describe in the conclusions.

\section{Preliminaries}\label{sec:prelim}

This section introduces basic definitions and properties that will be used in later sections.

\subsection{Datasets, Itemsets, and Association Rules}\label{sec:preldm}
A
\emph{dataset} $\Ds$ is a collection of \emph{transactions}, where each
transaction $\tau$ is a subset of a ground set $\Itm$\footnote{We assume
$\Itm=\cup_{\tau\in\Ds} \tau$, i.e., all the elements of $\Itm$ appear in at
least one transaction from $\Ds$.} There can be multiple
identical transactions in $\Ds$. Elements of $\Itm$ are called \emph{items} and
subsets of $\Itm$ are called $\emph{itemsets}$. Let $|\tau|$ denote the number
of items in transaction $\tau$, which we call the \emph{length} of $\tau$. Given
an itemset $A\subseteq\Itm$, the \emph{support set} of $A$, denoted as
$T_\Ds(A)$, is the set of transactions in $\Ds$ that contain $A$. The
\emph{support} of $A$, $s_\Ds(A) =|T_\Ds(A)|$, is the number of transaction
in $\Ds$ that contains $A$, and the \emph{frequency} of $A$, $f_\Ds(A)=
|T_\Ds(A)|/|\Ds|$, is the fraction of transactions in $\Ds$ that contain
$A$.
\begin{definition}\label{def:minethreshold}
  Given a \emph{minimum frequency
  threshold} $\theta$, $0<\theta\le 1$, the \emph{FI's mining task with respect
  to $\theta$} is finding all itemsets with frequency $\geq\theta$, i.e., the
  set 
  \[ \FI(\Ds,\Itm,\theta)=\left\{(A,f_\Ds(A)) ~:~ A \subseteq\Itm \mbox{ and }
  f_{\Ds}(A)\ge \theta\right\}.  \]
\end{definition}
  To define the collection of top-$K$ FI's, we assume a fixed \textit{canonical
  ordering} of the itemsets in $2^\Itm$ by decreasing frequency in $\Ds$, with
  ties broken arbitrarily, and label the itemsets $A_1,A_2,\dotsc,A_m$ according
  to this ordering.  For a given $1 \leq K \leq m$, we denote by
  $f^{(K)}_\Ds$ the frequency $f_\Ds(A_K)$ of the $K$-th most frequent itemset
  $A_K$, and define the set of top-$K$ FI's  (with their respective frequencies)
  as \[
  \TOPK(\Ds,\Itm,K)= \FI\left(\Ds,\Itm,f^{(K)}_\Ds\right).  \]

One of the main uses of frequent itemsets is in the discovery of
association rules.
%\begin{definition}\label{def:ar}
An \emph{association rule} $W$ is an expression
  ``$A\Rightarrow B$'' where $A$ and $B$ are itemsets such that $A\cap
  B=\emptyset$. The \emph{support} $s_\Ds(W)$ (resp.~frequency $f_\Ds(W)$)
  of the association rule $W$ is the support (resp.~frequency) of the itemset
  $A\cup B$. The \emph{confidence} $c_\Ds(W)$ of $W$ is the ratio $f_\Ds(A \cup
  B)/f_\Ds(A)$. %of the frequency of $A\cup B$ to the frequency of $A$.
%\end{definition}
Intuitively, an association rule ``$A\Rightarrow B$'' expresses, through its
support and confidence, how likely it is for the itemset $B$ to appear in the
same transactions as itemset $A$. The confidence of the association rule
can be interpreted the conditional probability of $B$ being present in a transaction that 
contains $A$.

\begin{definition}\label{def:minear}
  Given a dataset $\Ds$ with transactions
  built on a ground set $\Itm$, and given a minimum frequency threshold $\theta$
  and a minimum confidence threshold $\gamma$, the \emph{AR's task with respect
  to $\theta$ and $\gamma$} is to identify the set
  \[
  \AR(\Ds,\Itm,\theta,\gamma)=\left\{(W,f_\Ds(W),c_\Ds(W)) ~|~ \mbox{association rule } W,
  f_\Ds(W)\ge\theta, c_\Ds(W)\ge\gamma\right\}.\]
\end{definition}

We say that an itemset $A$ (resp.~an
association rule $W$) is in $\FI(\Ds,\Itm,\theta)$ or in $\TOPK(\Ds,\Itm,K)$
(resp.~in $\AR(\Ds,\Itm,\theta,\gamma)$) 
%and denote this fact with $A\in\FI(\Ds,\Itm,\theta)$ or $A\in\TOPK(\Ds,\Itm,K)$ (resp.
%$W\in\AR(\Ds,\Itm,\theta,\gamma)$),
when there $A$ (resp.~$W$) is part of a pair in $\FI(\Ds,\Itm,\theta)$ or
$\TOPK(\Ds,\Itm,K)$, (resp.~ a triplet $\AR(\Ds,\Itm,\theta,\gamma)$).
%$(A,f)\in\FI(\Ds,\Itm,\theta)$ or $(A,f)\in\TOPK(\Ds,\Itm,K)$ (resp.~a triplet
%$(W,f_w,c_w)\in\AR(\Ds,\Itm,\theta,\gamma)$).

In this work we are interested in extracting absolute and relative
approximations of the sets $\FI(\Ds,\Itm,\theta)$, $\TOPK(\Ds,\Itm,K)$ and
$\AR(\Ds,\Itm,\theta,\gamma)$. 

\begin{definition}\label{def:approxfi}
  Given a parameter
  $\varepsilon_{\mathrm{abs}}$ (resp.~$\varepsilon_{\mathrm{rel}}$), an
  \emph{absolute $\varepsilon_{\mathrm{abs}}$-close approximation}  (resp.~a
  \emph{relative $\varepsilon_{\mathrm{rel}}$-close approximation}) of
  $\FI(\Ds,\Itm,\theta)$ is a set $\mathcal{C}=\{(A, f_A) ~:~ A\subseteq\Itm,
  f_A\in[0,1]\}$ of pairs $(A, f_A)$ where $f_A$ approximates $f_\Ds(A)$.
  $\mathcal{C}$ is such that:
  \begin{enumerate}
    \item $\mathcal{C}$ contains all itemsets appearing in $\FI(\Ds,\Itm,\theta)$;
    \item $\mathcal{C}$ contains no itemset $A$ with frequency $f_\Ds(A)<\theta -
      \varepsilon_{\mathrm{abs}}$ (resp. $f_\Ds(A)< (1-\varepsilon_{\mathrm{rel}})\theta$);
    \item For every pair $(A, f_A)\in\mathcal{C}$, it holds
      $|f_\Ds(A)-f_A|\le\varepsilon_{\mathrm{abs}}$ (resp.
      $|f_\Ds(A)-f_A|\le\varepsilon_\mathrm{rel}f_\Ds(A)$).
  \end{enumerate}
\end{definition}

This definition extends easily to the case of top-$K$ frequent itemsets mining
using the equivalence 
\[ \TOPK(\Ds,\Itm,K)=\FI\left(\Ds,\Itm,f^{(K)}_\Ds\right):\]
an absolute (resp.  relative) $\varepsilon$-close approximation to
$\FI\left(\Ds,\Itm,f^{(K)}_\Ds\right)$ is an absolute (resp. relative)
$\varepsilon$-close approximation to $\TOPK(\Ds,\Itm,K)$.

For the case of association rules, we have the following definition.

\begin{definition}\label{def:approxar} Given a parameter
  $\varepsilon_{\mathrm{abs}}$ (resp.~$\varepsilon_{\mathrm{rel}}$), an
  \emph{absolute $\varepsilon_{\mathrm{abs}}$-close approximation}  (resp.~a
  \emph{relative $\varepsilon_{\mathrm{rel}}$-close approximation}) of
  $\AR(\Ds,\Itm,\theta,\gamma)$ is a set 
  \[
  \mathcal{C}=\{(W, f_W, c_W) ~:~
  \mbox{association rule } W, f_W\in[0,1], c_W\in[0,1]\}\]
  of triplets $(W, f_W, c_W)$ where $f_W$ and $c_W$ approximate $f_\Ds(W)$ and $c_\Ds(W)$
  respectively. $\mathcal{C}$ is such that:
  \begin{enumerate} 
      \item
      $\mathcal{C}$ contains all association rules appearing in
      $\AR(\Ds,\Itm,\theta,\gamma)$; \item $\mathcal{C}$ contains no association
      rule $W$ with frequency $f_\Ds(W)<\theta-\varepsilon_{\mathrm{abs}}$
      (resp. $f_\Ds(W)< (1-\varepsilon_{\mathrm{rel}})\theta$);
    \item For every triplet $(W, f_W, c_W)\in\mathcal{C}$, it holds
      $|f_\Ds(W)-f_W|\le\varepsilon_\mathrm{abs}$ (resp.
      $|f_\Ds(W)-f_W|\le\varepsilon_\mathrm{rel}\theta$).
    \item $\mathcal{C}$ contains no association rule $W$ with confidence
      $c_\Ds(W)<\gamma-\varepsilon_{\mathrm{abs}}$ (resp. $c_\Ds(W)<
      (1-\varepsilon_{\mathrm{rel}})\gamma$);
    \item For every triplet $(W, f_W, c_W)\in\mathcal{C}$, it holds
      $|c_\Ds(W)-c_W|\le\varepsilon_\mathrm{abs}$ (resp.
      $|c_\Ds(W)-c_W|\le\varepsilon_\mathrm{rel}c_\Ds(W)$).
  \end{enumerate}
\end{definition}

Note that the definition of relative $\varepsilon$-close approximation to
$\FI(\Ds,\Itm,\theta)$ (resp. to $\AR(\Ds,\Itm,\theta,\gamma)$) is more
stringent than the definition of $\varepsilon$-close solution to frequent
itemset mining (resp. association rule mining)
in~\cite[Sect.~3]{ChakaravarthyPS09}. Specifically, we require an approximation
of the frequencies (and confidences) in addition to the approximation of
the collection of itemsets or association rules (Property 3 in Def.~\ref{def:approxfi} and properties
3 and 5 in Def.~\ref{def:approxar}).

\subsection{VC-Dimension}\label{sec:prelvcdim}
The Vapnik-Chernovenkis (VC) Dimension of a space of points is a measure of the
complexity or expressiveness of a family of indicator functions (or equivalently
a family of subsets) defined on that space~\citep{VapnikC71}. A finite bound on
the VC-dimension of a structure implies a bound on the number of random samples
required for approximately learning that structure. We outline here some basic
definitions and results and refer the reader to the works
of~\citet[Sec.~14.4]{AlonS08},~\citet{DevroyeGL96} and~\citet{Vapnik99} for more details
on VC-dimension. See Sec.~\ref{sec:prevwork} for applications of VC-dimension in 
computer science.

%VC-dimension is defined on {\em range spaces}:
%
%\begin{definition}\label{defn:rangespace}
We define a {\em range space} as a pair $(X,R)$ where $X$ is a (finite or infinite) set
 and $R$ is a (finite or infinite) family of subsets of $X$. The members of $X$
 are called {\em points} and those of $R$ are called {\em ranges}.
 %\end{definition}
%To define the VC-dimension of a range space we consider the projection of the
%ranges into a set of points:
%\begin{definition}\label{defn:proj}
%Let $(X,R)$ be a range space and $A\subset X$.
Given $A\subset X$, The {\em projection} of $R$ on
$A$ is defined as $P_R(A)=\{r\cap A ~:~ r\in R\}$.
%\end{definition}
%The VC-dimension is defined with respect to \emph{shattered} sets:
%
%\begin{definition}\label{defn:shatter}
%Let $(X,R)$ be a range space and $A\subset X$.
If $P_R(A)=2^A$, then $A$ is said to be {\em shattered by $R$}.
%\end{definition} 
The VC-dimension of a range space is the cardinality of the largest set
shattered by the space:
\begin{definition}\label{defn:VCdim}
  Let $S=(X,R)$ be a range space. The {\em Vapnik-Chervonenkis} dimension (or
  {\em VC-dimension}) of $S$, denoted as $\VC(S)$ is the maximum cardinality of
  a shattered subset of $X$. If there are arbitrary large shattered subsets,
  then $\VC(S)=\infty$.
\end{definition}

Note that a range space $(X,R)$ with an arbitrary large set of points $X$ and
an arbitrary large family of ranges $R$ can have a bounded VC-dimension. A simple
example is the family of intervals in $[0,1]$ (i.e. $X$ is all the points in
$[0,1]$ and $R$ all the intervals $[a,b]$, such that $0\leq a\leq b\leq 1$). Let
$A=\{x,y,z\}$ be the set of three points $0<x<y<z<1$. No interval in $R$ can
define the subset $\{x,z\}$ so the VC-dimension of this range space is less than
3~\cite[Lemma 10.3.1]{Matousek02}.
%\begin{lemma}\label{lem:matousek}
%  The VC-Dimension of the range space $(\mathbb{R}^d, X)$, where $X$ is the set
%  of all half-spaces in $\mathbb{R}^d$ equals $d+1$.
%\end{lemma}

The main application of VC-dimension in statistics and learning theory is its
relation to the size of the sample needed to approximate learning the ranges, in
the following sense.

\begin{definition}\label{defn:eapprox}
  Let $(X,R)$ be a range space and let $A$
  be a finite subset of $X$. For $0<\varepsilon<1$, a subset $B\subset A$ is an
  $\varepsilon${\em-approximation} for $A$ if for all $r\in R$, we have
      \begin{equation}\label{eq:defeapprox}
	\left|\frac{|A\cap r|}{|A|}-\frac{|B\cap r|}{|B|}\right| \leq
	\varepsilon.
      \end{equation}
\end{definition}

A similar definition offers relative guarantees.
\begin{definition}\label{defn:releapprox}
  Let $(X,R)$ be a range space and let $A$
  be a finite subset of $X$. For $0<p,\varepsilon<1$, a subset $B\subset A$ is a
  \emph{relative} $(p,\varepsilon)$\emph{-approximation} for $A$ if for any
  range $r\in R$ such that $|A\cap r|/|A|\geq p$ we have 
  \[ \left|\frac{|A\cap r|}{|A|}-\frac{|B\cap r|}{|B|}\right| \leq
  \varepsilon\frac{|A\cap r|}{|A|}
  \]
  and for any range $r\in R$ such that $|A\cap r|/|A|< p$ we have $|B\cap
	r|/|B| \leq (1+\varepsilon)p$.
\end{definition}

An $\varepsilon$-approximation (resp.~a relative
$(p,\varepsilon)$-approximation) can be constructed by random sampling points of
the point space~[\cite{HarPS11}, Thm.~2.12 (resp.~2.11), see
also~\citep{LiLS01}].

\begin{theorem}\label{thm:eapprox}
  There is an absolute positive constant $c$ (resp. $c'$)
  such that if $(X,R)$ is a range-space of VC-dimension at most $d$, $A\subset
  X$ is a finite subset and $0<\varepsilon,\delta<1$ (resp.~and $0<p<1$), then a
  random subset $B\subset A$ of cardinality $m$, where
  \begin{equation}\label{eq:eapprox}
    m\ge\min\left\{|A|,\frac{c}{\varepsilon^2}\left(d+\log\frac{1}{\delta}\right)\right\},
  \end{equation}
  (resp.
  $m\ge\min\left\{|A|,c'\varepsilon^{-2}p^{-1}\left(d\log1/p-\log1/\delta\right)\right\}$)
  is an $\varepsilon$-approximation (resp.~a relative
  $(p,\varepsilon)$-approximation) for $A$ with probability at least $1-\delta$.
\end{theorem}

Note that throughout the work we assume the sample to be drawn \emph{with}
replacement if $m<|A|$ (otherwise the sample is exactly the set $A$).  The
constants $c$ and $c'$ are absolute and do not depend on the range space or on
any other parameter. \citet{LofflerP09} showed experimentally that the absolute
constant $c$ is at most $0.5$. No upper bound is currently known for $c'$.
Up to a constant, the bounds presented in Thm.~\ref{thm:eapprox} are
tight~\cite[Thm.~5]{LiLS01}. 

It is also interesting to note that an
$\varepsilon$-approximation of size
$O(d\varepsilon^{-2}(\log d-\log\varepsilon))$ can be built
\emph{deterministically} in time
$O(d^{3d}(\varepsilon^{-2}(\log d-\log\varepsilon))^d|X|)$
\citep{Chazelle00}.

\section{The Dataset's Range Space and its VC-dimension}\label{sec:vcdimar}
Our next step is to define a range space of the dataset and the itemsets. We
will use this space together with Theorem~\ref{thm:eapprox} to compute the bounds to
sample sizes sufficient to compute approximate solutions for the various tasks
of market basket analysis. 

\begin{definition}
  Let $\Ds$ be a dataset of transactions that are subsets of a ground set
  $\Itm$.  We define $S=(X,R)$ to be a range space associated with $\Ds$ such
  that:
  \begin{enumerate}
    \item $X=\Ds$ is the set of transactions in the dataset.
    \item $R=\{T_\Ds(A) ~|~ A\subseteq \Itm, A\neq\emptyset\}$ is a family of
      sets of transactions such that for each non-empty itemset
      $A\subseteq\Itm$, the set $T_\Ds(A)=\{\tau\in\Ds ~|~ A\subseteq\tau\}$ of
      all transactions containing $A$ is an element of $R$.
  \end{enumerate}
\end{definition}

It is easy to see that in practice the collection $R$ of ranges contains all and only
the sets $T_\Ds(A)$ where $A$ is a \emph{closed itemset}, i.e., a set such that
for each non-empty $B\subseteq A$ we have $T_\Ds(B)=T_\Ds(A)$ and for any
$C\supset A$, $T_\Ds(C)\subsetneq T_\Ds(A)$. Closed itemsets are used to
summarize the collection of frequent itemsets~\citep{CaldersRB06}.

The VC-Dimension of this range space is the maximum size of a set of
transactions that can be shattered by the support sets of the itemsets, as
expressed by the following theorem and the following corollary.

\begin{theorem}
  Let $\Ds$ be a dataset and let $S=(X,R)$ be the associated range
  space. Let $d\in\mathbb{N}$. Then $\VC(S)\ge d$ if and only if there exists a
  set $\mathcal{A}\subseteq\Ds$ of $d$ transactions from $\Ds$ such that for
  each subset $\mathcal{B}\subseteq\mathcal{A}$, there exists an itemset
  $I_\mathcal{B}$ such that the support set
  of $I_\mathcal{B}$ in $\mathcal{A}$ is exactly $\mathcal{B}$, that is
  $T_\mathcal{A}(I_\mathcal{B})=\mathcal{B}$.
  %\begin{enumerate}
  %  \item all transactions in $\mathcal{B}$ contain $I_\mathcal{B}$, and
  %  \item no transaction $\rho\in\mathcal{A}\setminus\mathcal{B}$ contains
  %    $I_\mathcal{B}$.
  %\end{enumerate}
\end{theorem}

\begin{proof} ``$\Leftarrow$''.
  From the definition of $I_\mathcal{B}$, we have that
  $T_\Ds(I_\mathcal{B})\cap\mathcal{A}=\mathcal{B}$. By definition of
  $P_R(\mathcal{A})$ this means that $\mathcal{B}\in P_R(\mathcal{A})$, for any
  subset $\mathcal{B}$ of $\mathcal{A}$. Then $P_R(\mathcal{A})=2^\mathcal{A}$,
  which implies $\VC(S)\ge d$.

  ``$\Rightarrow$''. Let $\VC(S)\ge d$. Then by the definition of VC-Dimension there
  is a set $\mathcal{A}\subseteq\Ds$ of $d$ transactions from $\mathcal{D}$ such
  that $P_R(A)=2^{\mathcal{A}}$. By definition of $P_R(\mathcal{A})$, this means
  that for each subset $\mathcal{B}\subseteq\mathcal{A}$ there exists an itemset
  $I_\mathcal{B}$ such that $T_\Ds(I_\mathcal{B})\cap\mathcal{A}=\mathcal{B}$.
  We want to show that no transaction $\rho\in\mathcal{A}\setminus\mathcal{B}$
  contains $I_\mathcal{B}$. Assume now by contradiction that there is a
  transaction $\rho^*\in\mathcal{A}\setminus\mathcal{B}$ containing
  $I_\mathcal{B}$. Then $\rho^*\in T_\Ds(I_\mathcal{B})$ and, given that
  $\rho^*\in\mathcal{A}$, we have $\rho^*\in
  T_\Ds(I_\mathcal{B})\cap\mathcal{A}$. But by construction, we have that
  $T_\Ds(I_\mathcal{B})\cap\mathcal{A}=\mathcal{B}$ and
  $\rho^*\notin\mathcal{B}$ because $\rho^*\in\mathcal{A}\setminus\mathcal{B}$.
  Then we have a contradiction, and there can not be such a transaction
  $\rho^*$.
\end{proof}

\begin{corollary} Let $\Ds$ be a dataset and $S=(\Ds,R)$ be the corresponding
  range space. Then, the VC-Dimension $\VC(S)$ of $S$, is the maximum integer
  $d$ such that there is a set $\mathcal{A}\subseteq\Ds$ of $d$ transactions
  from $\Ds$ such that for each subset $\mathcal{B}\subseteq\mathcal{A}$ of
  $\mathcal{A}$, there exists an itemset $I_\mathcal{B}$ such that the support
  of $I_\mathcal{B}$ in $\mathcal{A}$ is exactly $\mathcal{B}$, that is
  $T_\mathcal{A}(I_\mathcal{B})=\mathcal{B}$.
  %\begin{enumerate}
  %  \item all transactions in $\mathcal{B}$ contain $I_\mathcal{B}$, and
  %  \item no transaction $\rho\in\mathcal{A}\setminus\mathcal{B}$ contains
  %    $I_\mathcal{B}$.
  %\end{enumerate}
\end{corollary}

For example, consider the dataset $\Ds=\{\{a,b,c,d\},\{a,b\},\{a,c\},\{d\}\}$ of
four transactions built on the set of items $\Itm=\{a,b,c,d\}$. It is easy to see
that the set of transactions $\mathcal{A}=\{\{a,b\},\{a,c\}\}$ can be shattered:
$\mathcal{A}=\mathcal{A}\cap T_\Ds(\{a\})$, $\{\{a,b\}\}=\mathcal{A}\cap
T_\Ds(\{a,b\})$, $\{\{a,c\}\}=\mathcal{A}\cap T_\Ds(\{a,c\})$,
$\emptyset=\mathcal{A}\cap T_\Ds(\{d\})$. It should be clear that there is no
set of three transactions in $\Ds$ that can be shattered, so the VC-dimension of
the range space associated to $\Ds$ is exactly two.

Computing the exact VC-dimension of the range space associated to a dataset is
extremely expensive from a computational point of view. This does not come as a
suprise, as it is known that computing the VC-dimension of a range space $(X,R)$
can take time $O(|R||X|^{\log|R|})$~\cite[Thm.~4.1]{LinialMR91}. It is instead
possible to give an upper bound to the VC-dimension, and a procedure to
efficiently compute the bound.

We now define a characteristic quantity of the dataset, called the
\emph{d-index} and show that it is a tight bound to the VC-dimension of the
range space associated to the dataset, then present an algorithm to efficiently
compute an upper bound to the d-index with a single linear scan of the dataset.

\begin{definition}\label{defn:dindex}
  Let $\Ds$ be a dataset. The \emph{d-index} of $\Ds$ is the maximum integer $d$
  such that $\Ds$ contains at least $d$ different transactions of length at
  least $d$ such that no one of them is a subset of another,
  % for any two of these transactions $t,s$ we have neither $s\subseteq t$
  %nor $t\subseteq s$, 
  that is, the transactions form an \emph{anti-chain}.
\end{definition}
%\begin{definition}\label{defn:dindex}
%  Let $\Ds$ be a dataset and let $\Ds'$ be the set of transactions built by
%  removing from $\Ds$ all transactions that are equal to $\Itm$. Consider a
%  labelling $t_1,t_2,\dotsc,t_{|\Ds'|}$ of the transactions of $\Ds'$ so that
%  $|t_i|>|t_j|$ for all $1\le i < j \le |\Ds'|$, ties broken arbitrarily.
%  Consider now the anti-chain $\mathcal{G}=\{t_i ~|~ \neg\exists t_j \mbox{
%  s.t. } j<i \mbox{ and } t_j\supseteq t_i\}$. The \emph{d-index} of $\Ds$ is the
%  maximum integer $d$ such that $\mathcal{G}$ contains at least $d$ transactions
%  of length at least $d$.
%\end{definition}

Consider now the dataset $\Ds=\{\{a,b,c,d\},\{a,b,d\},\{a,c\},\{d\}\}$ of four
transactions built on the set of items $\Itm=\{a,b,c,d\}$. The d-index of $\Ds$ is
$2$, as the transactions $\{a,b,d\}$ and $\{a,c\}$ form an anti-chain.
Note that the anti-chain determining the d-index 
%as $\mathcal{G}=\{\{a,b\},\{a,c\},\{d\}\}$. Note that $\mathcal{G}$ 
is not necessarily the largest anti-chain that can be built on the transactions of
$\Ds$. For example, if
$\Ds=\{\{a,b,c,d\},\{a,b\},\{a,c\},\{a\},\{b\},\{c\},\{d\}\}$, the largest
anti-chain would be $\{\{a\},\{b\},\{c\},\{d\}\}$, but the anti-chain
determining the d-index of the dataset 
%$\mathcal{G}$
would be $\{\{a,b\},\{a,c\},\{d\}\}$.

Intuitively, the reason for considering an anti-chain of transactions is  
that, if $\tau$ is a transaction  that is a subset of another transaction
$\tau'$, ranges containing $\tau'$ necessarily also contain $\tau$ (the opposite
is not necessarily true), so it would be impossible to shatter a set containing
both transactions. 

%Note also that we can ignore, in the computation of the d-index,
%all transactions equal to $\Itm$ because they appear in all ranges $T_\Ds(A)$
%for all itemsets $A\subseteq\Itm$, so it is %impossible to shatter any set of
%transactions containing one or more such %transactions. Moreover, we can ignore

It is easy to see that the d-index of a dataset built on a set of items $\Itm$
is at most equal to the length of the longest transaction in the dataset and in
any case no greater than $|\Itm|-1$.

%The d-index can be viewed as a stricter version of the h-index measure for
%published work~\citep{Hirsch05}.
%Consider a dataset corresponding to the publications of one scientist.
%Each transaction corresponds to an article,
%and items in a transaction correspond to citations of that article.
%The h-index of this collection is $h$ if there are $h$ transactions of length at 
%least $h$, and there are no $h+1$ transactions of length $h+1$. The d-index adds
%an additional requirement, namely that among the $h$ transactions no transaction
%is a subset of another transaction. It can be equal the h-index, by instance
%when there are $k$ papers each with at least $k$ citations, and for each paper
%there is a citation that is unique for that paper.  

The d-index is an upper bound to the VC-dimension of a dataset.

\begin{theorem}\label{lem:vcdimupperb}
  Let $\Ds$ be a dataset with d-index $d$. Then the range space $S=(X,R)$
  corresponding to $\Ds$ has VC-dimension at most $d$.
\end{theorem}

\begin{proof}
  Let $\ell>d$ and assume that $S$ has
  VC-dimension $\ell$. From Def.~\ref{defn:VCdim} there is a set $\mathcal{K}$ of $\ell$
  transactions of $\Ds$ that is shattered by $R$. Clearly, $\mathcal{K}$ cannot
  contain any transaction equal to $\Itm$, because such transaction would appear
  in all ranges of $R$ and so it would not be possible to shatter $\mathcal{K}$.
  At the same time, for any two transactions $\tau,\tau'$ in
  $\mathcal{K}$ we must have neither $\tau\subseteq\tau'$ nor
  $\tau'\subseteq\tau$, otherwise the shorter transaction of the two would appear in
  all ranges where the longer one appears, and so it would not be possible to
  shatter $\mathcal{K}$. Then $\mathcal{K}$ must be an anti-chain. From this and
  from the definitions of $d$ and $\ell$, $\mathcal{K}$ must contain a
  transaction $\tau$ such that $|\tau|\le d$. %Fix $\tau$. 
  The transaction
  $\tau$ is a member of $2^{\ell-1}$ subsets of $\mathcal{K}$. We denote these subsets of $\mathcal{K}$ containing $\tau$ as
  $\mathcal{A}_i$, $1\le i\le 2^{\ell-1}$, labeling them in an
  arbitrary order. Since $\mathcal{K}$ is shattered (i.e., $P_R(\mathcal{K})=2^\mathcal{K}$), we have
  \[ 
  \mathcal{A}_i\in P_R(\mathcal{K}), 1\le i\le 2^{\ell -1}.
  \]
  From the above and the definition of $P_R(\mathcal{K})$, it follows that for
  each set of transactions $\mathcal{A}_i$ there must be a
  non-empty itemset $B_i$ such that 
  \begin{equation}\label{eq:vcdimupperb}
  T_\Ds\left(B_i\right)\cap \mathcal{K}= \mathcal{A}_i \in P_R(\mathcal{K}).
  \end{equation}
  Since the $\mathcal{A}_i$ are all different from each other, this
  means that the $T_\Ds(B_i)$ are all different from each other, which
  in turn requires that the $B_i$ be all different from each other,
  for $1\le i\le 2^{\ell-1}$. 

  Since $\tau \in \mathcal{A}_i$ and $\tau \in \mathcal{K}$ by
  construction, it follows from \eqref{eq:vcdimupperb} that 
  \[
  \tau \in T_\Ds\left(B_i\right), 1\le i\le 2^{\ell-1}.
  \]
  From the above and the definition of $T_\Ds(B_i)$, we get that all the
  itemsets $B_i, 1\le i\le 2^{\ell-1}$ appear in the transaction
  $\tau$. But $|\tau|\le d < \ell$, therefore $\tau$ can only contain at most $2^d-1 <
  2^{\ell -1}$ non-empty itemsets, while there are $2^{\ell-1}$ different
  itemsets $B_i$.

  This is a contradiction, therefore our assumption is false and
  $\mathcal{K}$ cannot be shattered by $R$, which implies that $\VC(S)\le d$.
\end{proof}

This bound is strict, i.e., there are indeed datasets with VC-dimension exactly
$d$, as formalized by the following Theorem.

\begin{theorem}\label{lem:vcdimlowerb}
  There exists a dataset $\Ds$ with d-index $d$ and such the corresponding range
  space has VC-dimension exactly $d$.
\end{theorem}

\begin{proof}
  For $d=1$, $\Ds$ can be any dataset with at least two different transactions
  $\tau=\{a\}$ and $\tau'=\{b\}$ of length 1. The set $\{\tau\}\subseteq\Ds$ is
  shattered because $T_\Ds(\{a\})\cap\{\tau\}=\{\tau\}$ and
  $T_\Ds(\{b\})\cap\{\tau\}=\emptyset$.

  Without loss of generality, let the ground set $\Itm$ be $\mathbb{N}$. For a
  fixed $d>1$, let $\tau_i=\{0,1,2,\dots,i-1,i+1,\dots,d\}$
  $1\le i\le d$, and consider the set of $d$ transactions $\mathcal{K}=\{\tau_i,
  1\le i\le d\}$.  Note that $|\tau_i|=d$ and $|\mathcal{K}|=d$ and for no pair
  of transactions $\tau_i,\tau_j$ with $i\neq j$ we have either
  $\tau_i\subseteq\tau_j$ nor $\tau_j\subseteq\tau_i$.
  
  $\Ds$ is a dataset containing $\mathcal{K}$ and any number of arbitrary
  transactions from $2^\Itm$ of length at most $d$. Let $S=(X,R)$ be the range
  space corresponding to $\Ds$. We now show that $\mathcal{K}\subseteq X$ is
  shattered by ranges from $R$, which implies
  $\VC(S)\ge d$. 
  
  For each $\mathcal{A}\in 2^\mathcal{K}\setminus\{\mathcal{K},\emptyset\}$, let
  $Y_\mathcal{A}$ be the itemset 
  \[ 
  Y_\mathcal{A}=\{1,\dots,d\}\setminus \{i ~:~
  \tau_i \in \mathcal{A}\}.
  \]
  Let $Y_\mathcal{K}=\{0\}$ and let $Y_\emptyset=\{d+1\}$. By construction we
  have
  \[
  T_\mathcal{K}(Y_\mathcal{A})=\mathcal{A}, \forall \mathcal{A}\subseteq\mathcal{K}
  \]
  i.e., the itemset $Y_\mathcal{A}$ appears in all transactions in $\mathcal{A}\subseteq \mathcal{K}$
  but not in any transaction from $\mathcal{K}\setminus\mathcal{A}$, for all $\mathcal{A}\in 2^{\mathcal{K}}$. This means that
  \[
  T_\Ds(Y_\mathcal{A})\cap \mathcal{K} = T_\mathcal{K}(Y_\mathcal{A}) =
  \mathcal{A}, \forall \mathcal{A}\subseteq\mathcal{K}.
  \]
  Since for all $\mathcal{A}\subseteq\mathcal{K}, T_\Ds(Y_\mathcal{A})\in R$ by
  construction, the above implies that
  \[
  \mathcal{A}\in P_R(\mathcal{K}), \forall \mathcal{A}\subseteq\mathcal{K}
  \]
  This means that $\mathcal{K}$ is shattered by $R$, hence $\VC(S)\ge d$. From
  this and Thm.~\ref{lem:vcdimupperb}, we can conclude that $\VC(S)=d$.
\end{proof}

Consider again the dataset $\Ds=\{\{a,b,c,d\},\{a,b\},\{a,c\},\{d\}\}$ of
four transactions built on the set of items $\Itm=\{a,b,c,d\}$. We argued before
that the VC-dimension of the range space associated to this dataset is exactly
two, and it is easy to see that the d-index of $\Ds$ is also two. 

\subsection{Computing the d-index of a dataset}
%The d-index $d$ of a dataset $\Ds$ can be computed in one scan of the dataset and
%with total memory $O(d)$. The scan starts with $d^*=1$ and it stores the length of
%the first transaction. At any given step the procedure stores $d^*$, the current
%estimate of $d$, computed as the maximum $d'$ such that the scan up
%to this step found at least $d'$ transactions with length at least $d'$, and
%keeps a list of the sizes of the transactions longer than $d'$ found so
%far. There can be no more than $d'$ such transactions. As the scan proceeds, the
%procedure updates $d^*$ and the list of transactions sizes greater than $d^*$.
The d-index of a dataset $\Ds$ %based on Def.~\ref{defn:dindex}
exactly can be obtained in polynomial time by computing, for each length $\ell$, the
size $w_\ell$ of the largest
anti-chain that can be built using the transactions of length at least $\ell$
from $\Ds$. If $w\ge\ell$, then the d-index is at least $\ell$. The maximum $\ell$
for which $w_\ell\ge\ell$ is the d-index of $\Ds$. The size of the largest
anti-chain that can be built on the elements of a set can be computed by solving
a maximum matching problem on a bipartite graph that has two nodes for each
element of the set~\citep{FordF62}. Computing the maximum matching can be done in
polynomial time~\citep{HopcroftK73}.

%The pseudocode for this procedure is presented in Alg.~\ref{alg:dindex}.
In practice, this approach can be quite slow as it requires, for each value
taken by $\ell$, a scan of the dataset to create the set of transactions of
length at least $\ell$, and to solve a maximum matching problem.
%It requires sorting the datasets by decreasing transaction length, which is not
%a viable option in most cases.
Hence, we now present an algorithm to efficiently compute an upper bound $q$ to the
d-index with a single linear scan of the dataset and with $O(q)$ memory. 
%Intuitively,
%it is easy to see that it is possible to compute an upper-bound to the d-index
%very easily with a single scan of the dataset using memory $O(q)$ by scanning

\begin{algorithm}[htb]
  \SetKwInOut{Input}{Input}
  \SetKwInOut{Output}{Output}
  \SetKwFunction{GetNextTransaction}{getNextTransaction}
  \SetKwFunction{ScanIsNotComplete}{scanIsNotComplete}
  \SetKw{Continue}{continue}
  \SetKw{MyAnd}{and}
  \DontPrintSemicolon
  %\dontprintsemicolon
  \Input{a dataset $\Ds$}
  \Output{an upper bound to the d-index of $\Ds$}
  $\tau\leftarrow$ \GetNextTransaction{$\Ds$}\;
  $\mathcal{T}\leftarrow\{\tau\}$\;
  $q\leftarrow 1$\;
  \While{\ScanIsNotComplete} {
  $\tau\leftarrow$ \GetNextTransaction{$\Ds$}\;
  \If {$|\tau|> q$ \MyAnd $\tau\neq\Itm$ \MyAnd $\neg\exists a\in \mathcal{T}$ such that $\tau=a$} {
  $\mathcal{R}\leftarrow \mathcal{T}\cup\{\tau\}$\;
  $q \leftarrow$ max integer such that $\mathcal{R}$ contains at least $q$ transactions
    of length at least $q$\;
    $\mathcal{T}\leftarrow$ set of the $q$ longest transactions from $\mathcal{R}$ (ties broken
    arbitrarily)\;
  }
  }
  \Return $q$\;
  \caption{Compute an upper bound to the d-index of a
  dataset}\label{alg:dindexbound}
\end{algorithm}

It is easy to see that the d-index of a dataset $\Ds$ is upper bounded by the
maximum integer $q$ such that $\Ds$ contains at least $q$ different (that is not
containing the same items) transactions of length at least $q$ and less than
$|\Itm|$. This upper bound ignores the constraint that the transactions that
concur to the computation of the d-index must form an anti-chain. We can compute
this upper bound in a greedy fashion by scanning the dataset one transaction at
a time and keep in memory the maximum integer $q$ such that we saw at least $q$
transactions of length $q$ until this point of the scanning. We also keep in
memory the $q$ longest different transactions, to avoid counting transactions that
are equal to ones we have already seen because, as we already argued, a set
containing identical transactions can not be shattered and copies of a
transaction should not be included in the computation of the d-index, so it is
not useful to include them in the computation of the bound. The pseudocode for
computing the upper bound to the d-index in the way we just described is
presented in Algorithm~\ref{alg:dindexbound}. The following lemma deals with the
correctness of the algorithm.

\begin{lemma}\label{lem:algocorrect}
  The algorithm presented in Algorithm~\ref{alg:dindexbound} computes the maximum
  integer $q$ such that $\Ds$ contains at least $q$ different transactions of
  length at least $q$ and less than $|\Itm|$.
\end{lemma}

\begin{proof}
  The algorithm maintains the following invariant after each update of
  $\mathcal{T}$: the set $\mathcal{T}$ contains the $\ell$ longest (ties broken
  arbitrarily) different transactions of length at least $\ell$, where $\ell$ is
  the maximum integer $r$ for which, up until to this point of the scan, the
  algorithm saw at least $r$ different transactions of length at least $r$. It
  should be clear that if the invariant holds after the scanning is completed,
  the thesis follows because the return value $q$ is exactly the size
  $|\mathcal{T}|=\ell$ after the last transaction has been read and processed.
  
  It is easy to see that this invariant is true after the first transaction has
  been scanned. Suppose now that the invariant is true at the beginning of the $n+1$-th
  iteration of the while loop, for any $n$, $0\le n\le |\Ds|-1$. 
  %after the first $n$ iterations of the while loop.
  We want to show that it will still be true at the end of the $n+1$-th
  iteration. Let $\tau$ be the transaction examined at the $n+1$-th iteration of
  the loop. If $\tau=|\Itm|$, the invariant is still true at the end of the $n+1$-th
  iteration, as $\ell$ does not change and neither does $\mathcal{T}$ because
  the test of the condition on line 6 of Algorithm~\ref{alg:dindexbound} fails.
  The same holds if $|\tau|<\ell$. Consider now
  the case $|\tau|>\ell$. If $\mathcal{T}$ contained, at the beginning of the
  $n+1$-th iteration, one transaction equal to $\tau$, then clearly $\ell$ would not
  change and neither does $\mathcal{T}$, so the invariant is still true at the
  end of the $n+1$-th iteration. Suppose
  now that $|\tau|>\ell$ and that $\mathcal{T}$ did not contain any transaction
  equal to $\tau$. Let $\ell_i$ be, for $i=1,\dotsc,|\Ds|-1$, the value
  of $\ell$ at the end of the $i$-th iteration, and let $\ell_0=1$. If
  $\mathcal{T}$ contained, at the beginning of the $n+1$-th iteration, zero
  transactions of length $\ell_n$, then necessarily it contained $\ell_n$
  transactions of length greater than
  $\ell_n$, by our assumption that the invariant was true at the end of the
  $n$-th iteration. Since $|\tau|>\ell_n$, it follows that
  $\mathcal{R}=\mathcal{T}\cup\{\tau\}$ contains $\ell_n+1$ transactions of size
  at least $\ell_n+1$, hence the algorithm at the end of the $n+1$-th iteration has
  seen $\ell_n+1$ transactions of length at least $\ell_n+1$, so
  $\ell=\ell_{n+1}=\ell_n+1$. This implies that at the end of iteration $n+1$
  the set $\mathcal{T}$ must have size $\ell_{n+1}=\ell_n+1$, i.e., must contain
  one transaction more than at the beginning of the $n+1$-th iteration. This is
  indeed the case as the value $q$ computed on line 8 of
  Algorithm~\ref{alg:dindexbound} is exactly $|\mathcal{R}|=\ell_n+1$
  because of what we just argued about $\mathcal{R}$, and therefore
  $\mathcal{T}$ is exactly $\mathcal{R}$ at the end of the $n+1$-th iteration and
  contains the $\ell=\ell_{n+1}$ longest different transactions of length at
  least $\ell$, which is exactly what is expressed by the invariant. If instead
  $\mathcal{T}$ contained, at the beginning of the $n+1$-th iteration, one or more
  transactions of length $\ell_n$, then $\mathcal{T}$ contains at most
  $\ell_n-1$ transactions of length greater than $\ell_n$, and $\mathcal{R}$
  contains at most $\ell_n$ transactions of length at least $\ell_n+1$, hence
  $q=\ell_n$. This also means that the algorithm has seen, before the beginning
  of the $n+1$-th iteration, at most $\ell_n-1$ different transactions strictly longer
  than $\ell_n$. Hence, after seeing $\tau$, the algorithm has seen at most
  $\ell_n$ transactions of length at least $\ell_n+1$,  so at the end of the
  $n+1$-th iteration we will have $\ell=\ell_{n+1}=\ell_n$. This means that the
  size of $\mathcal{T}$ at the end of the $n+1$-th iteration is the same as it
  was at the beginning of the same iteration. This is indeed the case because of
  what we argued about $q$. At the end of the $n+1$-th iteration, $\mathcal{T}$
  contains 1) all transactions of length greater than $\ell_n$ that it already
  contained at the end of the $n$-th iteration, and 2) the transaction $\tau$, and
  3) all but one the transactions of length $\ell_n$ that it contained at the
  end of the $n$-th iteration. Hence the invariant is true at the end of the
  $n+1$-th iteration because $\ell$ did not change and we replaced in
  $\mathcal{T}$ a transaction of length $\ell_n$ with a longer transaction,
  that is, $\tau$. Consider now the case of $|\tau|=\ell$. Clearly if there is
  a transaction in $\mathcal{T}$ that is equal to $\tau$, the invariant is still
  true at the end of the $n+1$-th iteration, as $\ell$ does not change and
  $\mathcal{T}$ stays the same. If $\mathcal{T}$ did not contain, at the
  beginning of the $n+1$-th iteration, any transaction equal to $\tau$, then also in this case
  $\ell$ would not change, that is $\ell=\ell_{n+1}=\ell_n$, because by
  definition of $\ell$ the algorithm already saw at least $\ell$ different
  transactions of length at least $\ell$. This implies that $\mathcal{T}$ must
  have, at the end of the $n+1$-th iteration, the same size that it had at the
  beginning of the $n+1$-th iteration. This is indeed the case because the set $\mathcal{R}$ 
  contains $\ell+1$ different transactions of size at least $\ell$, but there is
  no value $b>\ell$ for which $\mathcal{R}$ contains $b$ transactions of length
  at least $b$, because of what we argued about $\ell$, hence
  $|\mathcal{T}|=q=\ell$. At the end of the $n+1$-th iteration the set
  $\mathcal{T}$ contains 1) all the transactions of length greater than $\ell$ that
  it contained at the beginning of the $n+1$-th iteration, and 2) enough transactions of
  length $\ell$ to make $|\mathcal{T}|=\ell$. This means that $\mathcal{T}$ can
  contain, at the end of the $n+1$-th iteration, exactly the same set of
  transactions that it contained at the beginning $n+1$-th iteration and since, 
  as we argued, $\ell$ does not change, then the invariant is still true at the
  end of the $n+1$-th iteration. This completes our proof that the invariant
  still holds at the end of the $n+1$ iteration for any $n$, and therefore holds
  at the end of the algorithm, proving the thesis.
\end{proof}

%To partially add it back, we keep track of
%the set $T$ of $q$ transactions while scanning the dataset and do not include in the
%computation of $q$ any transaction that is a subset of the $q$ transactions we
%are keeping in memory. It should be clear that, on the other hand, we consider,
%in the computation of $q$, transactions that are supersets of those in $T$,
%hence the upper bound. 
The fact that the computation of this upper bound can be easily performed in a with a
single linear scan of the dataset in an online greedy fashion makes it extremely
practical also for updating the bound when the dataset grows, assuming that the
added transactions are appended at the tail of the dataset.

%\begin{algorithm}[htb]
%  \SetKwInOut{Input}{Input}
%  \SetKwInOut{Output}{Output}
%  \SetKwFunction{GetNextTransaction}{getNextTransaction}
%  \SetKwFunction{ScanIsNotComplete}{scanIsNotComplete}
%  \SetKw{Continue}{continue}
%  \SetKw{Break}{break}
%  \SetKw{MyAnd}{and}
%  \DontPrintSemicolon
%    %\dontprintsemicolon
%  \Input{a dataset $\Ds$}
%  \Output{the d-index of $\Ds$}
%
%  Sort $\Ds$ by decreasing transaction length .\;
%  $\tau\leftarrow$ \GetNextTransaction{$\Ds$}\;
%  $T\leftarrow\{\tau\}$\;
%  $q\leftarrow 1$\;
%  \While{\ScanIsNotComplete} {
%  $\tau\leftarrow$ \GetNextTransaction{$\Ds$}\;
%  \If {$|\tau|\le q$} {
%  \Break\;
%  }
%  \If {$|\tau|> q$ \MyAnd $\tau\neq\Itm$ \MyAnd $\neg\exists a\in T$ such that $\tau\subseteq a$} {
%    $T\leftarrow T\cup\{\tau\}$\;
%    $q \leftarrow$ max integer such that $T$ contains at least $q$ transactions
%    of length at least $q$\;
%    $T\leftarrow$ set of the $q$ longest transactions from $T$ (ties broken
%    arbitrarily)\;
%  }
%  }
%  \Return $q$\;
%  \caption{Compute the d-index of a dataset}\label{alg:dindex}
%\end{algorithm}

\subsection{Connection with monotone monomials}
There is an interesting connection between itemsets built on a ground set
$\Itm$ and 
the class of \emph{monotone monomials on $|\Itm|$ literals}. 
A \emph{monotone} monomial is a conjunction of literals with no negations. The
class $\textsc{monotone-monomials}_{|\Itm|}$ is the class of
all monotone monomials on $|\Itm|$ Boolean variables, including the constant functions
{\bf 0} and {\bf 1}. The VC-Dimension of the range space
\[
(\{0,1\}^{|\Itm|},\textsc{monotone-monomials}_{|\Itm|})\]
is exactly $|\Itm|$~\cite[Coroll.~3]{NatschlagerS96}. It is easy to see that it is always
possible to build a bijective map between the itemsets in $2^\Itm$ and the
elements of $\textsc{monotone-monomials}_{|\Itm|}$ and that transactions built
on the items in $\Itm$ correspond to points of $\{0,1\}^{|\Itm|}$. This implies
that a dataset $\Ds$ can be seen as a sample from $\{0,1\}^{|\Itm|}$.

Solving the problems we are interested in by using the VC-Dimension $|\Itm|$ of
monotone-monomials as an upper bound to the VC-dimension of the itemsets would
have resulted in a much larger sample size than what is sufficient, given that
$|\Itm|$ can be much larger than the d-index of a dataset. Instead, the
VC-dimension of the range space $(\Ds,R)$ associated to a dataset $\Ds$ is
equivalent to the VC-dimension of the range space
$(\Ds,\textsc{monotone-monomials}_{|\Itm|})$, which is the \emph{empirical
VC-Dimension} of the range space
$(\{0,1\}^{|\Itm|},\textsc{monotone-monomials}_{|\Itm|})$ measured on $\Ds$. Our
results, therefore, also show a tight bound to the empirical VC-Dimension of the
class of monotone monomials on $|\Itm|$ variables.

\section{Mining (top-$K$) Frequent Itemsets and Association Rules}\label{sec:approx}
We apply the VC-dimension results to constructing efficient sampling algorithms
with performance guarantees for approximating the collections of FI's, top-K FI's and AR's.

\subsection{Mining Frequent Itemsets}\label{sec:absapproxfi}
We construct bounds for the sample size needed to obtain
relative/absolute $\varepsilon$-close approximations to the collection of FI's. The
algorithms to compute the approximations use a standard
exact FI's mining algorithm on the sample, with an appropriately adjusted minimum
frequency threshold, as formalized in the following lemma.

\begin{lemma}\label{lem:absapproxfi}
  Let $\Ds$ be a dataset with transactions built on a ground set $\Itm$, and let
  $d$ be the d-index of $\Ds$. Let
  $0<\varepsilon,\delta<1$. Let $\Sam$ be a random sample of $\Ds$ with size 
  \[
  |\Sam|=\min\left\{|\Ds|,\frac{4c}{\varepsilon^2}\left(d+\log\frac{1}{\delta}\right)\right\},\]
  for some absolute constant $c$. Then $\FI(\Sam,\Itm,\theta-\varepsilon/2)$ is an absolute
  $\varepsilon$-close approximation to $\FI(\Ds,\Itm,\theta)$ with probability
  at least $1-\delta$.
\end{lemma}

\begin{proof}
  Suppose that $\Sam$ is a $\varepsilon/2$-approximation of the range space $(X,R)$ corresponding
  to $\Ds$. From Thm.~\ref{thm:eapprox} we know that this happens with probability at least $1-\delta$.  
  This means that for all $X\subseteq\Itm$,
  $f_\Sam(X)\in[f_\Ds(X)-\varepsilon/2,f_\Ds(X)+\varepsilon/2]$.
  This holds in particular for the itemsets in
  $\mathcal{C}=\FI(\Sam,\Itm,\theta-\varepsilon/2)$, which therefore satisfies
  Property 3 from Def.~\ref{def:approxfi}. It also means that for all $X\in\FI(\Ds,\Itm,\theta),
  f_\Sam(X)\ge \theta-\varepsilon/2$, so $\mathcal{C}$ also guarantees Property
  1 from Def.~\ref{def:approxfi}. Let now $Y\subseteq\Itm$ be such that
  $f_\Ds(Y)< \theta-\varepsilon$. Then, for the properties of $\Sam$,
  $f_\Sam(Y)<\theta-\varepsilon/2$, i.e., $Y\notin \mathcal{C}$, which allows us
  to conclude that $\mathcal{C}$ also has Property 2 from Def.~\ref{def:approxfi}.
\end{proof}

We stress again that here and in the following theorems, the constant $c$ is
absolute and does not depend on $\Ds$ or on $d$, $\varepsilon$ or $\delta$.

One very interesting consequence of this result is that we do not need to know in advance
the minimum frequency threshold $\theta$ in order to build the sample: the
properties of the $\varepsilon$-approximation allow to use the same sample
for any threshold and for different thresholds, i.e., the sample does not need
to be rebuilt if we want to mine it with a threshold $\theta$ first and with
another threshold $\theta'$ later.

It is important to note that the VC-dimension of a dataset, and therefore the
sample size from~\eqref{eq:eapprox} needed to probabilistically obtain an
$\varepsilon$-approximation, is independent from the size (number of
transactions) in $\Ds$ and also of the size of $\FI(\Sam,\Itm,\theta)$. It
only depends on the quantity $d$, which is always less or equal to the length
of the longest transaction in the dataset, which in turn is less or equal to the
number of different items $|\Itm|$.

To obtain a relative $\varepsilon$-close approximation, we
need to add a dependency on $\theta$ as shown in the following Lemma.

\begin{lemma}\label{lem:relapproxfi}
  Let $\Ds$, $d$, $\varepsilon$, and $\delta$ as in Lemma~\ref{lem:absapproxfi}. Let
  $\Sam$ be a random sample of $\Ds$ with size 
  \[
  |\Sam| =
  \min\left\{|\Ds|,\frac{4(2+\varepsilon)c}{\varepsilon^2\theta(2-\varepsilon)}\left(d\log\frac{2+\varepsilon}{\theta(2-\varepsilon)}+\log\frac{1}{\delta}\right)\right\},\]
  for some absolute absolute constant $c$. Then $\FI(\Sam,\Itm,(1-\varepsilon/2)\theta)$ is a relative
  $\varepsilon$-close approximation to $\FI(\Ds,\Itm,\theta)$ with probability
  at least $1-\delta$.
\end{lemma}

\begin{proof}
  Let $p=\theta(2-\varepsilon)/(2+\varepsilon)$. From
  Thm.~\ref{thm:eapprox}, the sample $\Sam$ is a relative
  $(p,\varepsilon/2)$-approximation of the range space associated to $\Ds$ with
  probability at least $1-\delta$. For any itemset $X$ in
  $\FI(\Ds,\Itm,\theta)$, we have $f_\Ds(X)\ge\theta>p$, so
  $f_\Sam(X)\ge (1-\varepsilon/2)f_\Ds(X)\ge(1-\varepsilon/2)\theta$, which
  implies $X\in\FI(\Sam,\Itm,(1-\varepsilon/2)\theta))$, so Property 1
  from Def.~\ref{def:approxfi} holds. Let now $X$ be an itemsets with
  $f_\Ds(X)< (1-\varepsilon)\theta$. From our choice of $p$, we always have
  $p>(1-\varepsilon)\theta$, so $f_\Sam(X)\le p(1+\varepsilon/2) <
  \theta(1-\varepsilon/2)$. This means
  $X\notin\FI(\Sam,\Itm,(1-\varepsilon/2)\theta))$, as requested by
  Property 2 from Def.~\ref{def:approxfi}. 
  Since $(1-\varepsilon/2)\theta=p(1+\varepsilon/2)$, it follows
  that only itemsets $X$ with $f_\Ds(X)\ge p$ can be in
  $\FI(\Sam,\Itm,(1-\varepsilon/2)\theta))$. For these itemsets it holds
  $|f_\Sam(X)-f_\Ds(X)|\le f_\Ds(X)\varepsilon/2$, as requested by
  Property 3 from Def.\ref{def:approxfi}.
\end{proof}

\subsection{Mining Top-$K$ Frequent Itemsets}\label{sec:miningtopk}
Given the equivalence
$\TOPK(\Ds,\Itm,K)=\FI(\Ds,\Itm,f^{(K)}_\Ds)$, we could use the above
FI's sampling algorithms if we had a good approximation of $f^{(K)}_\Ds$, the
threshold frequency of the top-$K$ FI's.

For the absolute $\varepsilon$-close approximation we first execute a standard
top-$K$ FI's mining algorithm on the sample to estimate $f^{(K)}_\Ds$ and then
run a standard FI's mining algorithm on the same sample using a minimum frequency
threshold depending on our estimate of $f_\Sam^{(K)}$.
Lemma~\ref{lem:absapproxtopk} formalizes this intuition.

\begin{lemma}\label{lem:absapproxtopk}
  Let $\Ds$, $d$, $\varepsilon$, and $\delta$ be as in Lemma~\ref{lem:absapproxfi}.
  Let $K$ be a positive integer. Let $\Sam$ be a random sample of $\Ds$ with
  size
  \[
  |\Sam|=\min\left\{|\Ds|,\frac{16c}{\varepsilon^2}\left(d+\log\frac{1}{\delta}\right)\right\}\],
  for some absolute constant $c$, then $\FI(\Sam,\Itm,f_\Sam^{(K)}-\varepsilon/2)$ is an absolute
  $\varepsilon$-close approximation to $\TOPK(\Ds,\Itm,K)$ with probability at
  least $1-\delta$.
\end{lemma}

\begin{proof}
  Suppose that $\Sam$ is a $\varepsilon/4$-approximation of the range
  space $(X,R)$ corresponding to $\Ds$. From Thm.~\ref{thm:eapprox} we know that
  this happens with probability at least $1-\delta$. This means that for all
  $Y\subseteq\Itm$,
  $f_\Sam(Y)\in[f_\Ds(Y)-\varepsilon/4,f_\Ds(Y)+\varepsilon/4]$.
  Consider now $f_\Sam^{(K)}$, the frequency of the $K$-th most frequent itemset
  in the sample. Clearly, $f_\Sam^{(K)}\ge f_\Ds^{(K)}-\varepsilon/4$,
  because there are at least $K$ itemsets (for example any subset of size $K$ of
  $\TOPK(\Ds,\Itm,K)$) with frequency in the sample at least
  $f_\Ds^{(K)}-\varepsilon/4$. On the other hand $f_\Sam^{(K)}\le
  f_\Ds^{(K)}+\varepsilon/4$, because there cannot be $K$ itemsets with a
  frequency in the sample greater than $f_\Ds^{(K)}+\varepsilon/4$: only
  itemsets with frequency in the dataset strictly greater than $f_\Ds^{(K)}$ can
  have a frequency in the sample greater than
  $f_\Ds^{(K)}+\varepsilon/4$, and there are at most $K-1$ such
  itemsets. Let now $\eta=f_\Sam^{(K)}-\varepsilon/2$, and consider
  $\FI(\Sam,\Itm,\eta)$. We have $\eta\le f_\Ds^{(K)}-\varepsilon/4$, so
  for the properties of $\Sam$,
  $\TOPK(\Ds,\Itm,K)=\FI(\Ds,\Itm,f_\Ds^{(K)})\subseteq\FI(\Sam,\Itm,\eta)$,
  which then guarantees Property 1 from Def.~\ref{def:approxfi}. On
  the other hand, let $Y$ be an itemset such that
  $f_\Ds(Y)<f_\Ds^{(K)}-\varepsilon$. Then $f_\Sam(Y)<
  f_\Ds^{(K)}-3\varepsilon/4\le\eta$, so $Y\notin\FI(\Sam,\Itm,\eta)$,
  corresponding to Property 2 from Def.~\ref{def:approxfi}. Property 3 from
  Def.~\ref{def:approxfi} follows from the properties of $\Sam$.
\end{proof}

Note that as in the case of the sample size required
for an absolute $\varepsilon$-close approximation to $\FI(\Ds,\Itm,\theta)$, we do
not need to know $K$ in advance to compute the sample size for obtaining
an absolute $\varepsilon$-close approximation to $\TOPK(\Ds,\Itm,K)$.

Two different samples are needed for computing a relative
$\varepsilon$-close approximation to $\TOPK(\Ds,\Itm,K)$, the first one to compute a
lower bound to $f_\Ds^{(K)}$, the second to extract the approximation. Details
for this case are presented in Lemma~\ref{lem:relapproxtopk}.

\begin{lemma}\label{lem:relapproxtopk}
  Let $\Ds$, $d$, $\varepsilon$, and $\delta$ be as in Lemma~\ref{lem:absapproxfi}.
  Let $K$ be a positive integer. Let $\delta_1,\delta_2$ be two reals such that
  $(1-\delta_1)(1-\delta_2)\ge(1-\delta)$. Let $\Sam_1$ be a random sample of
  $\Ds$ with some size
  \[
  |\Sam_1|=\frac{\phi c}{\varepsilon^2}\left(d+\log\frac{1}{\delta_1}\right)\]
  for some $\phi>2\sqrt{2}/\varepsilon$ and some absolute constant $c$. If
  $f_{\Sam_1}^{(K)}\ge (2\sqrt{2})/(\varepsilon\phi)$, then let
  $p=(2-\varepsilon)\theta/(2+\varepsilon)$ and let $\Sam_2$ be
  a random sample of $\Ds$ of size 
  \[ |\Sam_2|=\min\left\{|\Ds|,
  \frac{4c}{\varepsilon^2p}(d\log\frac{1}{p} + \log\frac{1}{\delta})\right\}\]
  for some
  absolute constant $c$. Then
  $\FI(\Sam_2,\Itm,(1-\varepsilon/2)(f_{\Sam_1}^{(K)}-\varepsilon/\sqrt{2\phi}))$
  is a relative $\varepsilon$-close approximation to $\TOPK(\Ds,\Itm,K)$ with
  probability at least $1-\delta$.
\end{lemma}

\begin{proof}
  Assume that $\Sam_1$ is a $\varepsilon/\sqrt{2\phi}$-approximation for
  $\Ds$ and $\Sam_2$ is a relative $(p,\varepsilon/2)$-approximation for $\Ds$.
  The probability of these two events happening at the same time is at least
  $1-\delta$, from Thm.~\ref{thm:eapprox}.

  Following the steps of the proof of Lemma~\ref{lem:absapproxtopk} we can
  easily get that, from the properties of $\Sam_1$,
  \begin{equation}\label{eq:boundfk}
    f_{\Sam_1}^{(K)}-\frac{\varepsilon}{\sqrt{2\phi}}\le f_\Ds^{(K)}\le
    f_{\Sam_1}^{(K)}+\frac{\varepsilon}{\sqrt{2\phi}}.
  \end{equation}

  Consider now an element $X\in\TOPK(\Ds,\Itm,K)$. We have by definition
  $f_\Ds(X)\ge f_\Ds^{(K)} > f_{\Sam_1}^{(K)}-\varepsilon/\sqrt{2\phi}\ge
  p$, and from the properties of $\Sam_2$, it follows that $f_\Sam(X)\ge
  (1-\varepsilon/2)f_\Ds(X)\ge(1-\varepsilon/2)(f_{\Sam_1}^{(K)}-\varepsilon/\sqrt{2\phi})$,
  which implies
  $X\in\FI(\Sam_2,\Itm,(1-\varepsilon/2)(f_{\Sam_1}^{(K)}-\varepsilon/\sqrt{2\phi}))$
  and therefore Property 1 from Def.~\ref{def:approxfi} holds for
  $FI(\Sam_2,\Itm,\eta)$.
 
  Let now $Y$ be an itemset such that $f_\Ds(Y)<(1-\varepsilon)f_\Ds^{(K)}$.
  From our choice of $p$ we have that $f_\Ds(A)< p$. Then
  $f_{\Sam_2}(A)<(1+\varepsilon/2)p<(1-\varepsilon/2)(f_{\Sam_1}^{(K)}-\varepsilon/\sqrt{2\phi})$.
  Therefore, $Y\notin\FI(\Sam_2,\Itm,\eta)$ and Property 2 from
  Def.~\ref{def:approxfi} is guaranteed.

  Property 3 from Def.~\ref{def:approxfi} follows from~\eqref{eq:boundfk} and
  the properties of $\Sam_2$.
\end{proof}

\subsection{Mining Association Rules}\label{sec:miningar}
Our final theoretical contribution concerns the discovery of relative/absolute
approximations to $\AR(\Ds,\Itm,\theta,\eta)$ from a sample.
Lemma~\ref{lem:relapproxar} builds on a result
from~\cite[Sect.~5]{ChakaravarthyPS09} and covers the \emph{relative} case,
while Lemma~\ref{lem:absapproxar} deals with the \emph{absolute} one.

\begin{lemma}\label{lem:relapproxar}
  Let $0<\delta,\varepsilon,\theta,\gamma<1$,
  $\phi=\max\{3,2-\varepsilon+2\sqrt{1-\varepsilon}\}$,
  $\eta=\varepsilon/\phi$, and $p=\theta(1-\eta)/(1+\eta)$. Let
  $\Ds$ be a dataset with d-index $d$.
  Let $\Sam$ be a random sample of $\Ds$ of size 
  \[
  |\Sam|=\min\left\{|\Ds|,\frac{c}{\eta^2p}(d\log\frac{1}{p}+\log\frac{1}{\delta})\right\}\]
  for some absolute constant $c$. Then
  $\AR(\Sam,\Itm,(1-\eta)\theta,\gamma(1-\eta)/(1+\eta))$
  is a relative $\varepsilon$-close approximation to
  $\AR(\Ds,\Itm,\theta,\gamma)$ with probability at least $1-\delta$.
\end{lemma}

\begin{proof}
  Suppose $\Sam$ is a relative $(p,\eta)$-approximation for the range space
  corresponding to $\Ds$. From Thm.~\ref{thm:eapprox} we know this happens with
  probability at least $1-\delta$.

 Let $W\in\AR(\Ds,\Itm,\theta,\gamma)$ be the association rule ``$A\Rightarrow
 B$'', where $A$ and $B$ are itemsets. By definition $f_\Ds(W)=f_\Ds(A\cup
 B)\ge\theta> p$. From this and the properties of $\Sam$, we get
\[
 f_\Sam(W)=f_\Sam(A\cup B)\ge (1-\eta)f_\Ds(A\cup B)\ge (1-\eta)\theta.\] 

Note that, from the fact that $f_\Ds(W)=f_\Ds(A\cup B)\ge\theta$, it follows
that $f_\Ds(A),f_\Ds(B)\ge\theta> p$, for the anti-monotonicity Property of the
frequency of itemsets.

By definition, $c_\Ds(W)=f_\Ds(W)/f_\Ds(A)\ge\gamma$. Then,
 \[
 c_\Sam(W)=\frac{f_\Sam(W)}{f_\Sam(A)}\ge
 \frac{(1-\eta)f_\Ds(W)}{(1+\eta)f_\Ds(A)}\ge\frac{1-\eta}{1+\eta}\cdot\frac{f_\Ds(W)}{f_\Ds(A)}\ge\frac{1-\eta}{1+\eta}\gamma.\]
 It follows that
 $W\in\AR(\Sam,\Itm,(1-\eta)\theta,\gamma(1-\eta)/(1+\eta))$, hence
 Property 1 from Def.~\ref{def:approxar} is satisfied.

 Let now $Z$ be the association rule ``$C\Rightarrow D$'', such that
 $f_\Ds(Z)=f_\Ds(C\cup D)<(1-\varepsilon)\theta$. But from our definitions of
 $\eta$ and $p$, it follows that $f_\Ds(Z) < p < \theta$, hence $f_\Sam(Z) <
 (1+\eta)p < (1-\eta)\theta$, and therefore
 $Z\notin\AR(\Sam,\Itm,(1-\eta)\theta,\gamma(1-\eta)(1+\eta))$, as
 requested by Property 2 from Def.~\ref{def:approxar}.
 
 Consider now an association rule $Y=\mbox{``}E\Rightarrow F\mbox{''}$ such that
 $c_\Ds(Y)<(1-\varepsilon)\gamma$. Clearly, we are only concerned with $Y$ such
 that $f_\Ds(Y)\ge p$, otherwise we just showed that $Y$ can not be in
 $\AR(\Sam,\Itm,(1-\eta)\theta,\gamma(1-\eta)/(1+\eta))$. From this and the
 anti-monotonicity property, it follows that $f_\Ds(E),f_\Ds(F)\ge p$. Then,
 \[
  c_\Sam(Y)=\frac{f_\Sam(Y)}{f_\Sam(E)}\le\frac{(1+\eta)f_\Ds(Y)}{(1-\eta)f_\Ds(E)}
  <\frac{1+\eta}{1-\eta} (1-\varepsilon)\gamma < \frac{1-\eta}{1+\eta}\gamma,\]
where the last inequality follows from the fact that
$(1-\eta)^2>(1+\eta)(1-\varepsilon)$ for our choice of $\eta$. We can
conclude that
$Y\notin\AR(\Sam,\Itm,(1-\varepsilon)\theta,\gamma(1-\eta)/(1+\eta)\gamma)$ and
therefore Property 4 from Def.~\ref{def:approxar} holds.

  Properties 3 and 5 from Def.~\ref{def:approxar} follow from the above steps (i.e.,
  what association rules can be in the approximations), from the definition of
  $\phi$, and from the properties of $\Sam$.
\end{proof}

\begin{lemma}\label{lem:absapproxar}
Let $\Ds$, $d$, $\theta$, $\gamma$, $\varepsilon$, and $\delta$ be as in
Lemma~\ref{lem:relapproxar}
and let $\varepsilon_\mathrm{rel}=\varepsilon/\max\{\theta,\gamma\}$.

Fix
$\phi=\max\{3,2-\varepsilon_\mathrm{rel}+2\sqrt{1-\varepsilon_\mathrm{rel}}\}$,
$\eta=\varepsilon_\mathrm{rel}/\phi$,
and $p=\theta(1-\eta)/(1+\eta)$. Let $\Sam$ be a random sample of $\Ds$ of
size 
\[
|\Sam|=\min\left\{|\Ds|,\frac{c}{\eta^2p}(d\log\frac{1}{p}+\log\frac{1}{\delta})\right\}\]
for some absolute constant $c$. Then
$\AR(\Sam,\Itm,(1-\eta)\theta,\gamma(1-\eta)/(1+\eta))$ is an absolute
$\varepsilon$-close approximation to $\AR(\Ds,\Itm,\theta,\gamma)$.
\end{lemma}

\begin{proof}
  The thesis follows from Lemma~\ref{lem:relapproxar} by setting $\varepsilon$
  there to $\varepsilon_\mathrm{rel}$.
\end{proof}

Note that the sample size needed for absolute $\varepsilon$-close
approximations to $\AR(\Ds,\Itm,\theta,\gamma)$ depends on $\theta$ and
$\gamma$, which was not the case for absolute $\varepsilon$-close approximations
to $\FI(\Ds,\Itm,\theta)$ and $\TOPK(\Ds,\Itm,K)$.

\section{Experimental Evaluation}\label{sec:exp}
In this section we present an extensive experimental evaluation of
our methods to extract approximations of $\FI(\Ds,\Itm,\theta)$, $\TOPK(\Ds,\Itm,K)$, and
$\AR(\Ds,\Itm,\theta,\gamma)$.

Our first goal is to evaluate the \emph{quality} of the
approximations obtained using our methods, by comparing the experimental results 
to the analytical bounds. We also evaluate how strict the bounds are
 by testing whether the same quality of results can be
achieved at sample sizes smaller than those suggested by the theoretical analysis. 
We then show that our methods can significantly speed-up the mining process,
fulfilling the motivating promises of the use of sampling in the market basket
analysis tasks. Lastly, we compare the sample sizes from our results to the best
previous work~\citep{ChakaravarthyPS09}.

We tested our methods on both real and artificial datasets. The real datasets
come from the FIMI'04 repository (\url{http://fimi.ua.ac.be/data/}). Since most
of them have a moderately small size, we replicated their transactions a number
of times, with the only effect of increasing the size of the dataset but no
change in the distribution of the frequencies of the itemsets. The artificial
datasets were built such that their corresponding range spaces
have VC-dimension equal to the maximum transaction length, which is the
maximum possible as shown in Thm.~\ref{lem:vcdimupperb}. To create these
datasets, we followed the proof of Thm.~\ref{lem:vcdimlowerb} and used the
generator included in ARtool (\url{http://www.cs.umb.edu/~laur/ARtool/}), which
is similar to the one presented in~\citep{AgrawalS94}. The artificial datasets
had ten million transactions. We used the FP-Growth and
Apriori implementations in ARtool to extract frequent itemsets and association
rules. To compute an upper bound $d$ to the d-index of the dataset, we used
Algorithm~\ref{alg:dindexbound}.
In all our experiments we fixed $\delta=0.1$. In the experiments involving
absolute (resp.~relative) $\varepsilon$-close approximations we set
$\varepsilon=0.01$ (resp.~$\varepsilon=0.05$). The absolute constant $c$ was fixed to
$0.5$ as suggested by~\citep{LofflerP09}. This is reasonable because, again, $c$
does not depend in any way from $\Ds$, $d$, $\varepsilon$, $\delta$, or
characteristics of the collection of FI's or AR's. No upper bound is currently
known for $c'$ when computing the sizes for relative
$\varepsilon$-approximations. We used the same value $0.5$ for $c'$ and
found that it worked well in practice. For each dataset we selected a
range of minimum frequency thresholds and a set of values for $K$ when
extracting the top-$K$ frequent itemsets. For association rules discovery we set
the minimum confidence threshold $\gamma\in\{0.5, 0.75, 0.9\}$. For each
dataset and each combination of parameters we created random samples with size
as suggested by our theorems and with smaller sizes to evaluate the strictness
of the bounds.  
We measured, for each set of parameters, the \emph{absolute
frequency error} and the \emph{absolute confidence error}, defined 
as  the error
$|f_\Ds(X)-f_\Sam(X)|$ (resp.~$|c_\Ds(Y)-c_\Sam(Y)|$) for an itemset $X$
(resp.~an association rule $Y$) in the approximate collection extracted from sample $\Sam$.
When dealing with the
problem of extracting \emph{relative} $\varepsilon$-close approximations, we
defined the \emph{relative frequency error} to be the absolute
frequency error divided by the real frequency of the itemset and analogously for
the relative confidence error (dividing by the real confidence). In the figures
we plot the maximum and the average for these quantities, taken over all
itemsets or association rules in the output collection. In order to limit the
influence of a single sample, we computed and plot in the figures the maximum
(resp.~the average) of these quantities in three runs of our methods on three
different samples for each size.

The first important result of our experiments is that, for all problems (FI's,
top-$K$ FI's, AR's) , for every combination of parameters, and for every run, the
collection of itemsets or of association rules obtained using our methods always
satisfied the requirements to be an absolute or relative $\varepsilon$-close
approximation to the real collection. Thus in practice our methods indeed
achieve or exceed the theoretical guarantees for approximations of the
collections $\FI(\Ds,\Itm,\theta)$, $\TOPK(\Ds,\Itm,\theta)$, and
$\AR(\Ds,\Itm,\theta,\gamma)$.  
Given that the collections returned by our algorithms where always a superset of
the collections of interest or, in other words, that the \emph{recall} of the
collections we returned was always 1.0, we measured the \emph{precision} of the
returned collection. In all but one case this statistic was at least $0.9$ (out
of a maximum of $1.0$), suggesting relatively few false positives in the
collections output. In the remaining case (extracting FI's from the dataset
BMS-POS), the precision ranged between $0.59$ to $0.8$ (respectively for
$\theta=0.02$ and $\theta=0.04$). The
probability of including a FI or an AR which has frequency (or confidence, for
AR's) less than $\theta$ (or $\gamma$) but does not violate the properties of a
$\varepsilon$-close approximation,  and is therefore an
``acceptable'' false positive, depends on the distribution of the
real frequencies of the itemsets or AR's in the dataset around the frequency
threshold $\theta$ (more precisely, below it, within $\varepsilon$ or
$\varepsilon\theta$): if many patterns have a real frequency in this interval,
then it is highly probable that some of them will be included in the collections
given in output, driving precision down. Clearly this probability depends on the
number of patterns that have a real frequency close to $\theta$. Given that
usually the lower the frequency the higher the number of patterns with that
frequency, this implies that our methods may include more ``acceptable'' false
positives in the output at very low frequency thresholds. Once again, this
depends on the distribution of the frequencies and does not violate the
guarantees offered by our methods. It is possible to use the output of our
algorithms as a set of \emph{candidate patterns} which can be reduced to the
real exact output (i.e., with no false positives) with a single scan of the
dataset.

\begin{figure}[tp]
  \centering
  \subfloat[Absolute Itemset Frequency
  Error, BMS-POS dataset, $d=81$, $\theta=0.02$, $\varepsilon=0.01$,
  $\delta=0.1$]{\label{fig:BMS-POS-absFI}\includegraphics[width=0.49\textwidth]{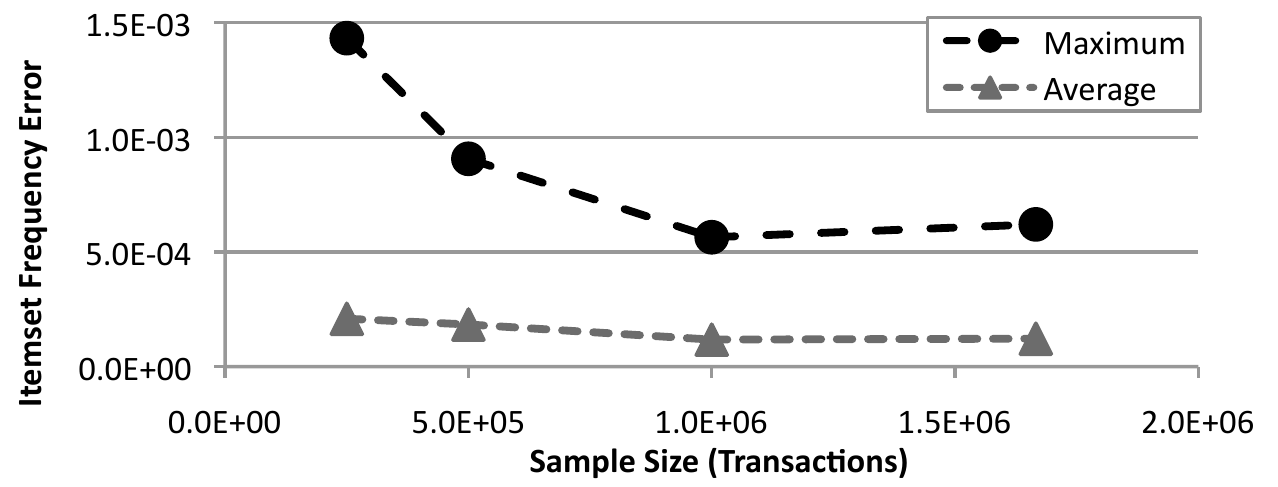}}
 \hfill
  \subfloat[Relative Itemset Frequency
  Error, artificial dataset, $d=33$, $\theta=0.01$, $\varepsilon=0.05$,
  $\delta=0.1$]{\label{fig:artif-relFI}\includegraphics[width=0.49\textwidth]{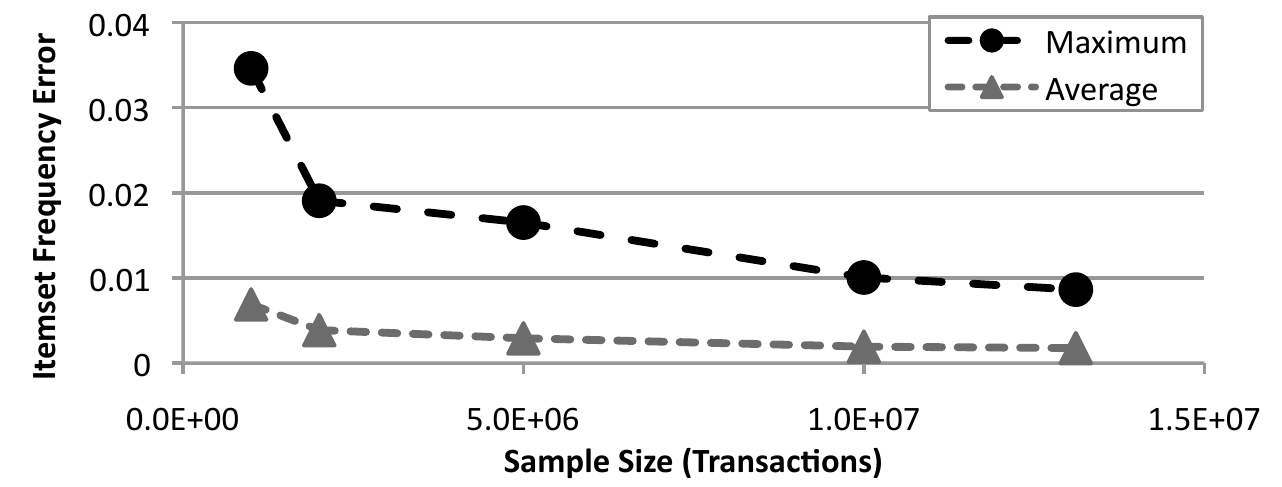}}
  \caption{Absolute / Relative $\varepsilon$-close Approximation to $\FI(\Ds,\Itm,\theta)$} 
  \label{fig:fiapprox}
\end{figure}

Evaluating the strictness of the bounds to the sample size was the second goal
of our experiments. In Figure~\ref{fig:BMS-POS-absFI} we show the behaviour of
the maximum frequency error as function of the sample size in the itemsets
obtained from samples using the method presented in Lemma~\ref{lem:absapproxfi}
(i.e., we are looking for an \emph{absolute} $\varepsilon$-close approximation
to $\FI(\Ds,\Itm,\theta)$). The rightmost plotted point corresponds to the
sample size suggested by the theoretical analysis. We are showing the results
for the dataset BMS-POS replicated $40$ times (d-index $d=81$), mined with
$\theta=0.02$. It is clear from the picture that the guaranteed error bounds are
achieved even at sample sizes smaller than what suggested by the analysis and
that the error at the sample size derived from the theory (rightmost plotted
point for each line) is one to two orders of magnitude smaller than the maximum tolerable
error $\varepsilon=0.01$. This fact seems to suggest that there is still room for
improvement in the bounds to the sample size needed to achieve an absolute
$\varepsilon$-close approximation to $\FI(\Ds,\Itm,\theta)$.
In Fig.~\ref{fig:artif-relFI}
we report similar results for the problem of computing a \emph{relative}
$\varepsilon$-close approximation to $\FI(\Ds,\Itm,\theta)$ for an artificial
dataset whose range space has VC-dimension $d$ equal to the length of the
longest transaction in the dataset, in this case $33$. The dataset contained 100
million transactions. The sample size, suggested by Lemma~\ref{lem:relapproxfi},
was computed using $\theta=0.01$, $\varepsilon=0.05$, and $\delta=0.1$.
The conclusions we can draw from the results for the behaviour of the
relative frequency error are similar to those we got for the absolute case.
For the case of absolute and relative $\varepsilon$-close approximation to
$\TOPK(\Ds,\Itm,K)$, we observed results very similar to those obtained for
$\FI(\Ds,\Itm,\theta)$, as it was expected, given the closed connection between
the two problems. 

\begin{figure}[tp]
  \centering
  \subfloat[Relative Association Rule Frequency
  Error]{\label{fig:artif-relAR-freq}\includegraphics[width=0.49\textwidth]{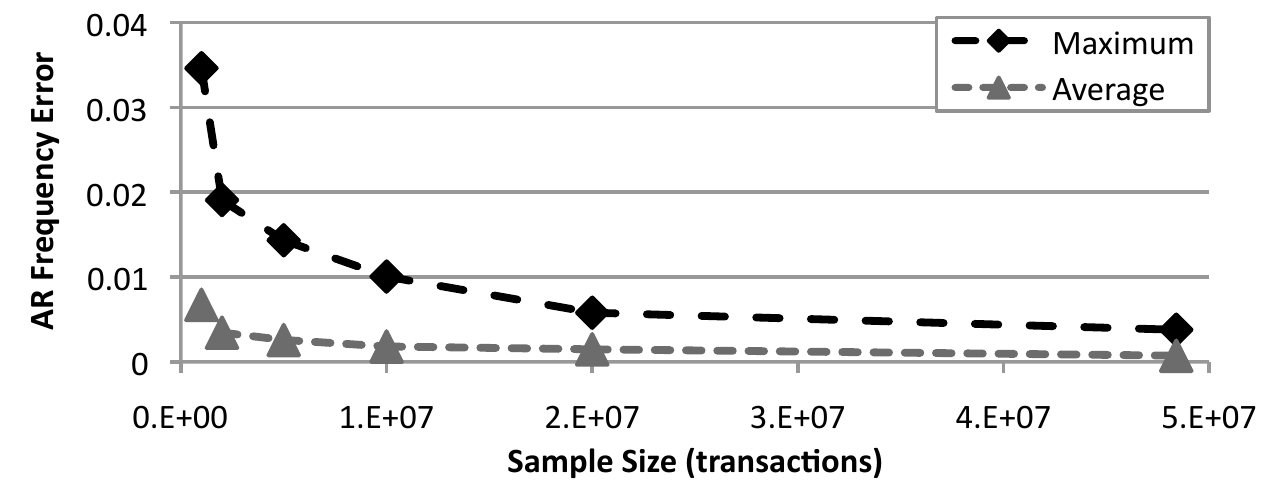}}
  \hfill
  \subfloat[Relative Association Rule Confidence 
  Error]{\label{fig:artif-relAR-conf}\includegraphics[width=0.49\textwidth]{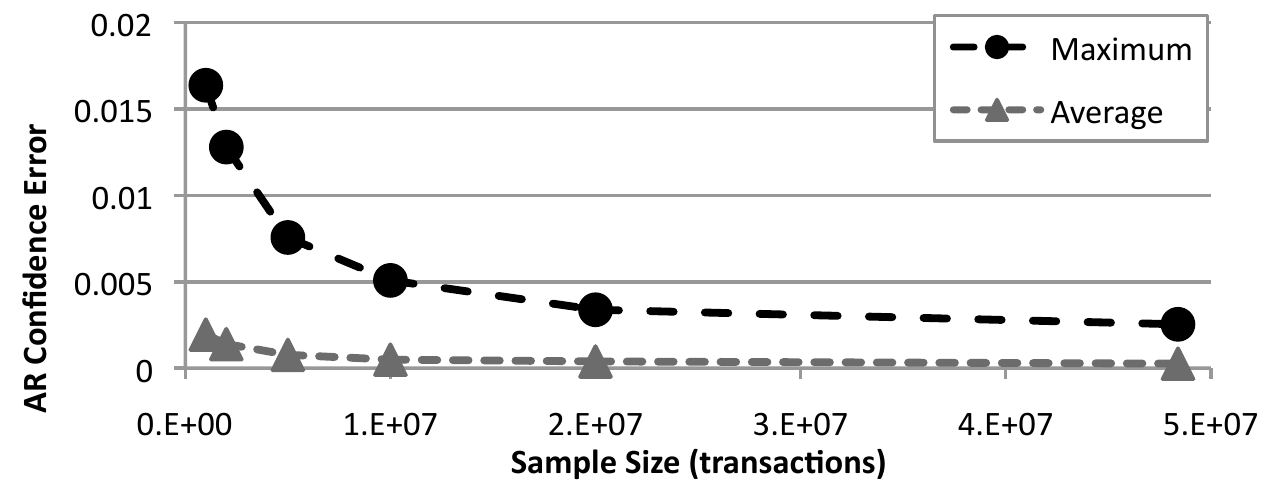}}
  \caption{ Relative $\varepsilon$-close approximation to
  $\AR(\Ds,\Itm,\theta,\gamma)$, artificial dataset, $d=33$, $\theta=0.01$,
  $\gamma=0.5$, $\varepsilon=0.05$, $\delta=0.1$.}
  \label{fig:ar}
\end{figure}

The results of the experiments to evaluate our method to extract a relative
$\varepsilon$-close approximation to $\AR(\Ds,\Itm,\theta,\gamma)$ are presented
in Fig.~\ref{fig:artif-relAR-freq}~and~\ref{fig:artif-relAR-conf}. The same
observations as before hold for the relative
frequency error, while it is interesting to note that the relative confidence
error is even smaller than the frequency error, most possibly because the
confidence of an association rule is the ratio between the frequencies of two
itemsets that appear in the same transactions and their sample frequencies will
therefore have similar errors that cancel out when the ratio is computed.
Similar conclusions can be made for the absolute $\varepsilon$-close
approximation case.

The major motivating intuition for the use of sampling in market basket analysis
tasks is that mining a sample of the dataset is faster than mining the entire
dataset. Nevertheless, the mining time does not only depend on the number of
transactions, but also on the number of frequent itemsets. Given that our
methods suggest to mine the sample at a lowered minimum frequency threshold,
this may cause an increase in running time that would make our method not useful
in practice, because there may be many more frequent itemsets than at the
original frequency threshold. We performed a number of experiments to evaluate
whether this was the case and present the results in Fig.~\ref{fig:runtime}. 
We mined the artificial dataset introduced before for different values of $\theta$,
and created samples of size sufficient to obtain a relative $\varepsilon$-close
approximation to $\FI(\Ds,\Itm,\theta)$, for $\varepsilon=0.05$ and
$\delta=0.1$. Figure~\ref{fig:runtime} shows the time needed to mine the large
dataset and the time needed to create and mine the samples. It is possible to
appreciate that, even considering the sampling time, the speed up achieved by
our method is around the order of magnitude (i.e. 10x speed improvement),
proving the usefulness of sampling. Moreover, given that the sample size, and
therefore the time needed to mine the sample, does not grow with the size of the
dataset as long as the upper bound to the d-index remains constant, that the
d-index computation can be performed online, and that the time to create the
sample can be made dependent only on the sample size using Vitter's Method D
algorithm~\citep{Vitter87}, our method is very scalable as the dataset grows, and
the speed up becomes even more relevant because the mining time for the large
dataset would instead grow with the size of the dataset.

\begin{figure}[tp]
  \centering
  \includegraphics[width=0.49\textwidth]{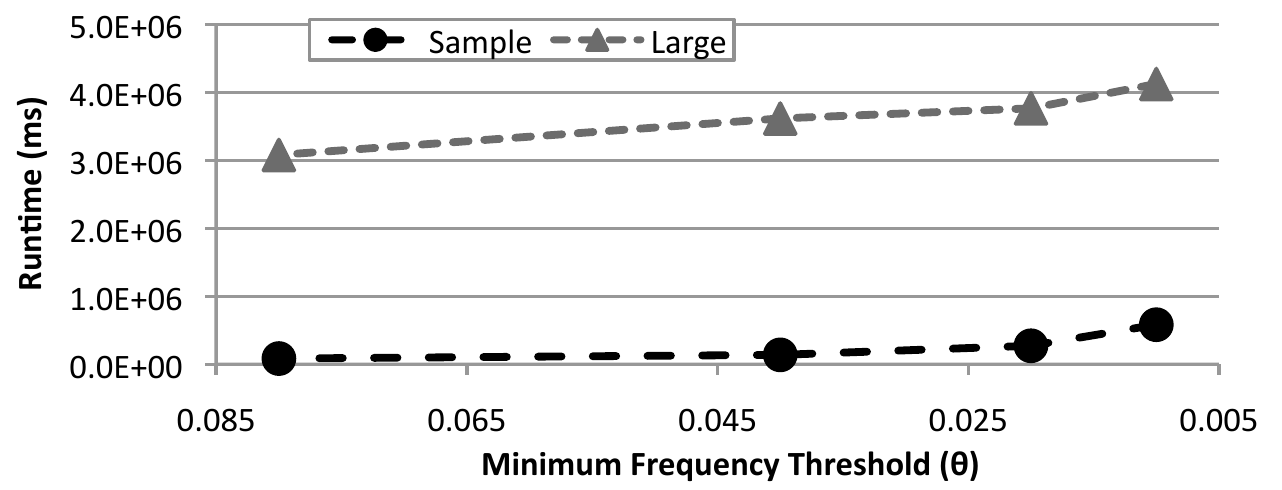}
  \caption{Runtime Comparison. The sample line includes the sampling time.
  Relative approximation to FI's, artificial dataset, $d=33$, $\varepsilon0.05$,
  $\delta=0.1$}
  \label{fig:runtime}
\end{figure}

Comparing our results to previous work we note that the bounds generated by our
technique are always linear in the VC-dimension $d$ associated with the dataset.
As reported in Table~\ref{table:comparsamsizeform}, the best previous
work~\citep{ChakaravarthyPS09} presented bounds that are linear in the maximum
transaction length $\Delta$ for two of the six problems studied here.
Figures~\ref{fig:compareSamSizeTheta} and~\ref{fig:compareSamSizeEpsilon} shows
a comparison of the actual sample sizes for relative $\varepsilon$-close
approximations to $\FI(\Ds,\Itm,\theta)$ for as function of $\theta$ and
$\varepsilon$. To compute the points for these figures, we set $\Delta=d=50$,
corresponding to the worst possible case for our method, i.e., when the
VC-dimension of the range space associated to the dataset is exactly equal to
the maximum transaction length. We also fixed $\delta=0.05$ (the two methods
behave equally as $\delta$ changes). For Fig.~\ref{fig:compareSamSizeTheta}, we
fixed $\varepsilon=0.05$, while for Fig.~\ref{fig:compareSamSizeEpsilon} we
fixed $\theta=0.05$. From the Figures we can appreciate that both bounds have
similar, but not equal, dependencies on $\theta$ and $\varepsilon$. More
precisely the bound presented in this work is less dependent on $\varepsilon$
and only slightly more dependent on $\theta$. It also evident that the sample
sizes suggested by the bound presented in this work are always much smaller than
those presented in~\citep{ChakaravarthyPS09} (the vertical axis has logarithmic
scale). In this comparison we used $\Delta=d$, but almost all real datasets
we encountered have $d\ll\Delta$ as shown in
Table~\ref{tab:deltadrealds} which would result in a larger gap between the
sample sizes provided by the two methods.

\begin{figure}[tp]
  \centering
  \subfloat[Sample size as function of
  $\theta$,
  $\varepsilon=0.05$]{\label{fig:compareSamSizeTheta}\includegraphics[width=0.49\textwidth]{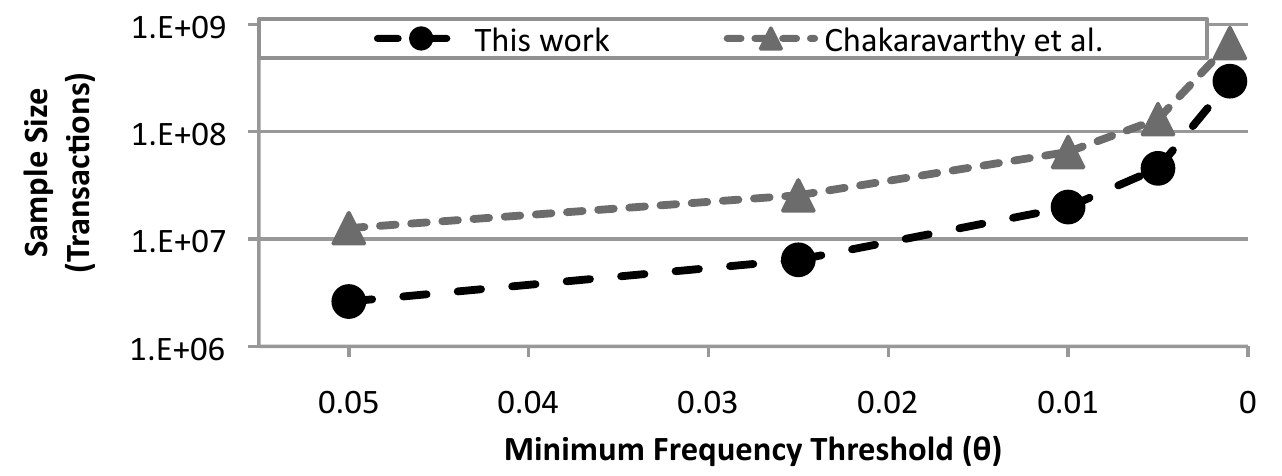}}
  \hfill
  \subfloat[Sample size as function of
  $\varepsilon$,
  $\theta=0.05$]{\label{fig:compareSamSizeEpsilon}\includegraphics[width=0.49\textwidth]{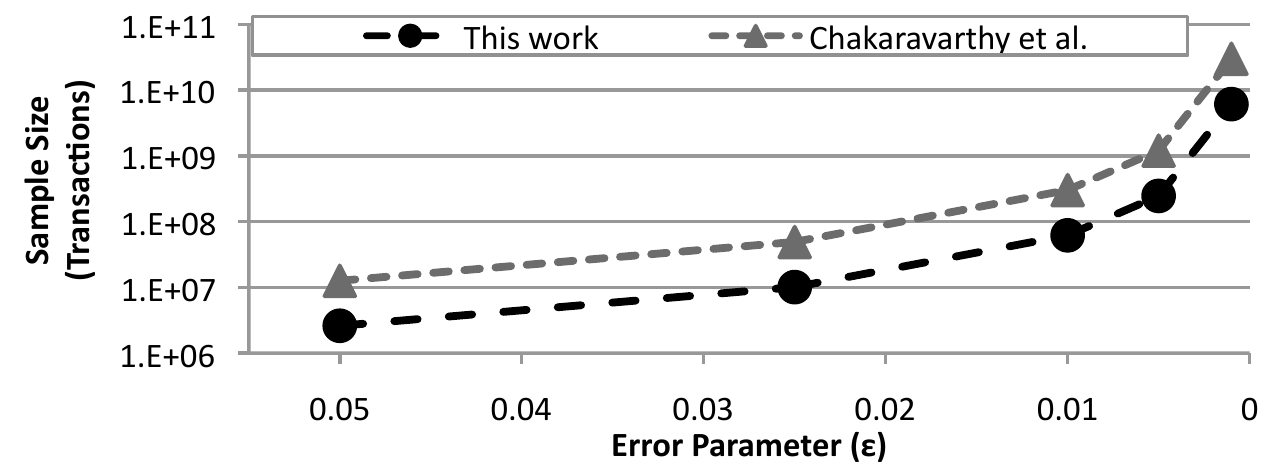}}
  \caption{Comparison of sample sizes for relative $\varepsilon$-close approximations to
  $\FI(\Ds,\Itm,\theta)$. $\Delta=d=50$, $\delta=0.05$.}
  \label{fig:compareSamSize}
\end{figure}

\begin{table}[hbt]
\centering
\caption{Values for maximum transaction length $\Delta$ and upper bound to the
d-index $d$ for real datasets}
\label{tab:deltadrealds}

\begin{tabular}{lcccccccc}
  \toprule
  & accidents & BMS-POS & BMS-Webview-1 & kosarak & pumsb* & retail & webdocs \\
  \midrule
  $\Delta$ & 51 & 164 & 267 & 2497 & 63 & 76 & 71472 \\
  $d$ & 46 & 81 & 57 & 443 & 59 & 58 & 2452 \\ 
  \bottomrule
\end{tabular}
\end{table}

\section{Conclusions}\label{sec:concl}
In this paper we presented a novel technique to derive random sample sizes
sufficient to easily extract high-quality approximations of the (top-$K$)
frequent itemsets and of the collection of association rules. The sample size
are linearly dependent on the VC-Dimension of the range space associated to the
dataset, which is upper bounded by the maximum integer $d$ such
that there at least $d$ transactions of length at least $d$ in the dataset. This 
bound is tight for a large family of datasets.  

We used theoretical tools from statistical learning theory to developed a very
practical solution to an important problem in computer science. The practicality
of our method is demonstrated in the extensive experimental evaluation which
confirming our theoretical analysis and suggests that in practice it is possible
to achieve even better results than what the theory guarantees. Moreover, we
used this method as the basic building block of an algorithm for the
MapReduce~\citep{DeanS04} distributed/parallel framework of computation.
PARMA~\citep{RiondatoDFU12}, our MapReduce algorithm, computes an absolute
$\varepsilon$-approximation of the collection of FI's or AR's by mining a number
of small random samples of the dataset in parallel and then aggregating and
filtering the collections of patterns that are frequent in the samples. It
allows to achieve very high-quality approximations of the collection of interest
with very high confidence while exploiting and adapting to the available
computational resources and achieving a high level of parallelism, highlighted 
by the quasi-linear speedup we measured while testing PARMA.

Samples of size as suggested by our methods can be used to mine approximations
of other collection of itemsets, provided that one correctly define the
approximation taking into account the guarantees on the estimation of the
frequency provided by the $\varepsilon$-approximation theorem. For example, one
can can use techniques like those presented in~\citep{MampaeyTV11} on a sample
to obtain a small collection of patterns that describe the dataset as best as
possible.

We believe that methods and tools developed in the context of computational
learning theory can be applied to many problems in data mining, and that results
traditionally considered of only theoretical interest can be used to obtain very
practical methods to solve important problems in knowledge discovery.

It may be possible to develop procedures that give a stricter upper bound to the
VC-dimension for a given dataset, or that other measures of sample complexity
like the triangular rank~\citep{NewmanR12}, shatter coefficients, or Rademacher
inequalities~\citep{BoucheronBL05}, can suggest smaller samples sizes. 

\paragraph{Acknowledgements.}
The authors are thankful to Luc De Raedt for suggesting the connection between
itemsets and monotone monomials.

% Bibliography
%\bibliographystyle{ACM-Reference-Format-Journals}
%\bibliography{fimine,vcmine,various,riondapubs,statsigfis}

\end{document}